\documentclass[12pt]{article}
\usepackage{enumerate}
\usepackage{amsmath}
\usepackage{amssymb}
\usepackage{bbm}
\usepackage{graphicx}
\usepackage{float}
\usepackage{subfigure}
\usepackage{hyperref}
\usepackage{color}
\def\E{\mathbb{E}}
\usepackage{natbib}
\usepackage{caption}
\newcommand{\pad}{\right)}
\newcommand{\pai}{\left(}
\newcommand{\cd}{\right]}
\newcommand{\ci}{\left[}

\newtheorem{lemma}{Lemma}
\newtheorem{remark}{Remark}
\newtheorem{propo}{Proposition}
\newtheorem{teo}{Theorem}

\newtheorem{defi}{Definition}

\newtheorem{hypo}{Hypothesis}
\newenvironment{proof}[1][Proof]{\textbf{#1.} }{\ \rule{0.5em}{0.5em}}
\newcommand{\rr}{\mathbb{R}}
\newcommand{\N}{\mathbb{N}}

\newcommand{\p}{\mathbb{P}}

\allowdisplaybreaks
\begin{document}
\title{$\,$\vskip-2cm\bf How fast extinction occurs in bisexual populations with size-depending mating dynamics?}
\author{{\sc Ehyter M. Mart\'in-Gonz\'alez}\thanks{Corresponding Author. Departamento de Matem\'aticas, División de Ciencias Naturales y Exactas, Universidad de Guanajuato,
  Guanajuato,  Mexico, ehyter.martin@ugto.mx}, \sc{Carlos Galv\'an-Galv\'an}\thanks{Departamento de Matem\'aticas, División de Ciencias Naturales y Exactas, Universidad de Guanajuato,
  Guanajuato,  Mexico, c.galvangalvan@ugto.mx},\\
  \sc{ Eduardo Calvo-Mart\'inez}\thanks{Departamento de Matem\'aticas, División de Ciencias Naturales y Exactas, Universidad de Guanajuato,
  Guanajuato,  Mexico, e.calvomartinez@ugto.mx}
}
\maketitle
\begin{abstract}
Given that extinction in a bisexual population is certain, we study a way to a\-pproximate the time when this extinction occurs. Our study is based on standard tools from Extreme Value Theory, which in practice are very easy to implement. We present the theoretical results derived from our study and provide a few numerical examples of such results.

\end{abstract}

\section*{Introduction}

Populations with sexual reproduction can be modelled using the class of stochastic processes known as Bisexual Galton-Watson processes (BGWP), which are a natural generalization of the class of Bienamé-Galton-Watson Branching Process, in which only asexual reproduction is considered. 
The BGWP was introduced by \cite{daley} and since then, it has been widely studied by several authors (see e.g. \cite{alsmeyerrosler}, \cite{bagley}, \cite{bruss}, \cite{hulletal} and the references therein).

In particular, some conditions for certain extinction (i.e. extinction with probability one) can be found in many works such as \cite{molinaetal1}). Once that it is known that extinction is certain for a given a population whose dynamics can be modelled using a BGWP, it is natural to ask about the probability that extinction occurs after (or before) a given time of interest. This kind of questions can be answered if we know the exact distribution of the time to extinction, which it is generally difficult to determine.

In \cite{hulletal} the authors present some bounds for the distribution of the time to extinction of a subcritical BGWBP (see (\ref{criticalidad})). This subcritical case corresponds, as in the case of an asexual Galton-Watson processes, to extinction with probability one. Such bounds are useful given that there is only one female and one male who mate and reproduce; otherwise, they may become greater than one, which is not informative.

In this work we focus on providing an alternative approach for the approximation of the probability that extinction occurs after a given time $k$. Our approach is based on statistical tools from Extreme Value Theory, which in practice are very easy to implement and take into account observed data (assuming that a suitable mating function has been previously chosen).

First we prove that even under extinction with probability one, assuming some conditions, the time to extinction is an unbounded random variable such that it takes large values with very small probability (i.e. a random variable with range on $(0,\infty)$ and a light tail).

Afterwards we show that the tail of the distribution (and hence the distribution) of the time to extinction can be approximated using an exponential distribution. To be precise, we show that the time to extinction $\tau$ is the discretization of a random variable whose distribution belongs to the Gumbel domain of attraction for maxima and hence, its distribution conditioned on the event when $\tau$ is greater than a sufficiently large threshold $u$ can be approximated using the Peaks Over Threshold method (POT).

This work is organized as follows: in Section 1 we formally introduce the BGWBP and present some preliminary concepts and facts about it. We also describe the conditions imposed on our process in order for the results to be valid. In Section 2 we state and prove our main result, as well as some technical lemmas . In Section 3 we briefly recall the way the Peaks Over Threshold method works and use it in some numerical examples to approximate the distribution of the time to extinction for some BGWBP that satisfy the conditions given in Section 2. We also show what happens to these same examples when the conditions in Section 2 are not fulfilled. Finally, in Section 4 we present our conclusions and some open problems related to this work.

\section{Preliminaries}

\subsection{Description of the model and assumptions}
We begin by introducing the standard bisexual Galton-Watson process (BGWP), $$\{(F_n,M_n,Z_n),n\in \N\cup\{0\}\}.$$ This process is defined as follows:
We have a population with two types of individuals, say females ($F$) and males ($M$) and a function $L:\rr\times\rr\to[0,\infty)$, which satisfies the following conditions:
\begin{enumerate}
 \item $L(x,y)=0$ for all $(x,y)\in \rr^2 \backslash (0,\infty)^2$,
 \item $L(x,y)\in \N$ if both $x,y\in \N$, 
 \item $L(x,y)\leq xy$ for $x,y\geq0$ and
 \item $L(x,y)$ is non decreasing in each component.
\end{enumerate}

This function is called \text{the mating function} of the BGWP. Typical examples of mating functions are

$$L(x,y)=\min\{x,y\}, L(x,y)=xy,L(x,y)=x,L(x,y)=y.$$

Note that the mating function is not assumed to be symmetric.

At the beginning of this process (generation zero) we have a total number of $F_0$ females and $M_0$ males. These individuals mate and form $Z_0:=L(F_0,M_0)$ mates according to the function $L$. Then each one of these $Z_0$ mates gives birth to a number of female and male children. For each $j\in\{1,\dots,Z_0\}$ we denote by $f_{j,0}$ and $m_{j,0}$ the number of female and male children from the $jth$-couple, where $\{f_{j,0},j\in\{1,\dots,Z_0\}\}$ are discrete iid random variables with $p_f:=\p\ci f_{j,0}=0\cd> 0$, $\{m_{j,0},j\in\{1,\dots,Z_0\}\}$ are discrete iid random variables with $p_m:=\p\ci m_{j,0}=0\cd> 0$. Furthermore, we also assume $\{f_{j,0},j\in\{1,\dots,Z_0\}\}$ are independent of $\{m_{j,0},j\in\{1,\dots,Z_0\}\}$.

Now, in generation 1 the total number of females and males are respectively given by

$$F_1=\sum\limits_{j=1}^{Z_0}f_{j,0},M_1=\sum\limits_{j=1}^{Z_0}m_{j,0}.$$

These $F_1$ females and $M_1$ males now form $Z_1:=L(F_1,M_1)$ mates. Each of these mates have $f_{j,1}$ and $m_{j,1}$ female and male children, respectively.

In general, the number of females and males at each generation $n$ are denoted respectively by $F_n$ and $M_n$, while the number of mates arising from these $F_n$ females and $M_n$ males is denoted by $Z_n:=L(F_n,M_n)$. The numbers of female and male children from the $jth$ mate (with $j\in\{1,\dots,Z_n\}$) are respectively denoted by $f_{j,n}$ and $m_{j,n}$. Hence, the total number of females and males in the $(n+1)$-generation are respectively given by 

$$F_{n+1}=\sum\limits_{j=1}^{Z_{n}}f_{j,n},M_n=\sum\limits_{j=1}^{Z_{n}}m_{j,n},$$
where the empty sum is defined as zero.

We make the assumption that $f_{j,n}\overset{d}{=}f_{k,m}$ and $m_{j,n}\overset{d}{=}m_{k,m}$ for all $n,m\in\N\cup\{0\}$ and $j,k\in\N$.

In this work we consider a generalization of the BGWBP described above. To be precise, we allow the mating function to also vary with $n$ in such way that we consider a sequence $\{L(n,x,y)\}_n$ of mating functions. This modification to the classical BGWBP was introduced in \cite{molinaetal1}. The idea here is to be able to model situations in which new mates are formed according to the mates in the previous generation, that is $Z_n$ depends not only on $F_{n-1}$ and $M_{n-1}$ but also on $Z_{n-1}$. A situation in which this dependency arises in real life is described in \cite{molinaetal3} (see Remark 2.1 and the comments below).

In this case, the mating function is a three-valued function

$$L:\N\times [0,\infty)\times [0,\infty)\to [0,\infty),$$

which will be denoted as $L_n(x,y):=L(n,x,y)$ for each natural number $n$.

Hence, the BGWBP is denoted as

$$\left\{(F_n,M_n,Z_n),n\in\N\cup\{0\}\right\},\text{ where } Z_{n+1}=L_{Z_{n}}(F_n,M_n).$$
\subsection{Subcritical BGWBP}

In what follows, $Z:=\{Z_n,n\in\N\}$ will be referred to as \emph{the mating process} of the BGWBP. Now we define the extinction time of $Z$ as

$$\tau:=\min\{k: Z_k=0\},$$

and set 
$$r_k:=\frac{1}{k}\E\ci Z_{n+1} | Z_n=k\cd, k\in\mathbb{N}\cup\{0\}.$$

We recall that a function $f:[0,\infty)^k\to [0,\infty)$ such that

$$f(x_{1,1}+x_{1,2},\dots,x_{k,1}+x_{k,2})\geq f(x_{1,1},\dots,x_{k,1})+f(x_{1,2},\dots,x_{k,2}),$$

is called \emph{superadditive} function. In \cite{molinaetal1}, the authors prove that if $L_n(x,y)$
is superadditive for each $n\in\mathbb{N}$ and for each $x,y\geq 0$, then $r:=\lim_{k\to\infty}r_k=\sup_{k\geq 1}r_k$ exists. Furthermore, they show that $r\leq 1$ is equivalent to extinction with probability one.

With this in mind, a BGWBP is said to be 
\begin{equation}
    \text{\textbf{subcritical} if $r<1$, \textbf{critical} if $r=1$ and \textbf{supercritical}
 if $r>1$}.\label{criticalidad}
\end{equation}
Throughout the rest of this work we assume the following conditions:

\begin{hypo}\label{hypo1}
    \leavevmode
    \begin{enumerate}
    \item The BGWBP is not trivial, i.e. $p_m,p_f\in(0,1)$.
        \item For each natural number $n$, the function $L_n(x,y)$ is superadditive.
        \item There exists positive constants $C,D$ such that for all $n$, either $L_n(x,y)\leq Cx$ or $L_n(x,y)\leq Dy$ holds.
        \item $r<1$.
        \item For each $j$, the probability $\p\ci f_{j,0}=0,m_{j,0}=0\cd=p_mp_f$ is stricly smaller than $1/2$. If we set $\theta:=-\ln{(p_mp_f)}$, this hypothesis becomes

        $$\theta >\ln{(2)}.$$
    \end{enumerate}
\end{hypo}

Hypothesis \ref{hypo1} part 3 is easily checked on functions such as $L_n(x,y)=\min\{x,g(n)y\}$ for any function $g\geq 0$, since $\min\{x,g(n)y\}\leq x$ and hence, the assumption holds for $x$ with $C=1$. 

Hypothesis 5 resulted from our simulation studies in order to avoid the case when $\tau$ behaves nearly as a degenerated random variable (see Section \ref{notheta}). At this point we have not been able to provide a formal proof that this condition implies that $\tau$ is not nearly a degenerated random variable. However, in all our simulation studies the assumption of such condition resulted in $\tau$ having a distribution with a wide range of values with positive probability, a condition needed to prove the second main result from this work.

To end this part of this section, we recall the following concepts which refer to classifications of distribution functions $G$ in terms of their tails $\overline{G}:=1-G$.

\begin{defi}\label{colaligera}A distribution function $G$ with $G(0)=0$ is said to be light-tailed if there exist positive and real constants $a,b$ such that

$$\overline{G}(x)\leq ae^{-bx},\quad \forall\ x>0.$$

We say that $G$ is heavy-tailed if $G$ is not light-tailed.
\end{defi}

The class of heavy-tailed distributions will be denoted by $\mathcal{H}$ and the class of light-tailed distributions will be denoted by $\mathcal{H}^\complement$.

\subsection{Preliminaries from Extreme Value Theory}

We say that a distribution $F$ belongs to the maximal domain of attraction of a distribution function $G$ if and only if there exist sequences of constants $\{a_n\}_n$ and $\{b_n\}_n$ such that $a_n>0$ for each $n$ and

$$\lim\limits_{n\to\infty}F^n(a_nx+b_n)=G(x),$$

for all $x$ which is a point of continuity for $G$.

The Fisher-Tippet theorem shows that the only distribution functions having a maximal domain of attraction coincide with the class of extreme value distributions (Fréchet, Gumbel and reversed Weibull distributions). For an extreme value distribution $H$ we define its maximal domain of attraction $D(H)$ as
\[D(H) := \left\{ F: \lim_{n\to\infty} F^n(a_nx + b_n) = H(x), a_n >0, b_n\in\rr\right\}.\]

It is known (see, e.g. , \cite{embrechtsetal} and the references therein) that $D(H_\text{Fréchet})$ contains only heavy-tailed distributions, $D(H_\text{Weibull})$ contains only distributions with finite right endpoint and $D(H_\text{Gumbel})$ includes some heavy-tailed and some light-tailed distributions.

\section{Main Results}

In what follows, $F_\tau$ will be the distribution function of $\tau$ under $\p$ and by $\overline{F}_\tau$ we denote the tail of $F_\tau$, i.e. $\overline{F}_\tau:=1-F_\tau$.

Let us first consider the following proposition concerning $\omega_\tau:= \sup\{x:F_\tau(x)<1\}$, i.e the right endpoint of the distribution $F_\tau$.

\begin{propo}\label{omegainfinito}
    For any non-trivial BGWBP with initial offspring $N\in\mathbb{N}$, it holds that $\omega_\tau=\infty$ (i.e. $\tau$ is an unbounded random variable).
\end{propo}
\begin{proof}
Following the same technique as in \cite{molinaetal1}, we set for $s\in [0,1]$ and $n\in\N$

$$g_n(s):=\E[s^{Z_n}],\quad h(s):=\E[s^{L_{1}(f_{1,1},m_{1,1})}].$$ 
$g_n$ is probability generating function of $Z_n$ for each $n$ and $h$ equals the probability generating function of $L_{1}(f_{1,1},m_{1,1})$. We also set 
\[
h^{(n)}(s):= \overbrace{h(h(\dots (s)\dots)}^{n}.
\]
As our first step in the proof of our result, we show that the following inequality holds for all $n\geq 2$:
\begin{equation}\label{cosa}
    g_n(s)\leq g_1(h^{(n-1)}(s)).
\end{equation}

Indeed, by applying the Law of Total Probability we obtain
\begin{align*}
  g_2(s)=\E[s^{Z_2}]&=\sum_{k=0}^\infty\E\ci s^{Z_2}|Z_1=k\cd\p[Z_1=k]=\sum_{k=0}^\infty \E\ci s^{L_{k}(\sum_{j=1}^kf_{j,1},\sum_{j=1}^km_{j,1})}\cd\p[Z_1=k], 
\end{align*}
and recalling the superaditivity of the mating functions, it follows that
\begin{align*}
    \sum_{k=0}^\infty \E\ci s^{L_{k}(\sum_{j=1}^kf_{j,1},\sum_{j=1}^km_{j,1})}\cd\p[Z_1=k]&\leq \sum_{k=0}^\infty \E\ci s^{\sum_{j=1}^kL_{1}(f_{j,1},m_{j,1})}\cd\p[Z_1=k]\\
    &=\sum_{k=0}^\infty\prod_{j=1}^k \E\ci s^{L_{1}(f_{j,1},m_{j,1})}\cd\p[Z_1=k] \\
    &=\sum_{k=0}^\infty \E\ci s^{L_{1}(f_{1,1},m_{1,1})}\cd^k\p[Z_1=k]\\
    &=g_1(h(s)).
\end{align*}

In a similar way we can show that
\[g_n(s)\leq g_{n-1}(h(s)),\quad \forall n\geq 3,
\]and the result follows by induction. 

    On the other hand, using that the BGWP is not trivial, for $s<t$ the following inequality holds:
    \begin{equation}\label{cosita}
        h(s)=\sum_{k=0}^\infty \p\ci L_{1}(f_{1,1},m_{1,1})=k\cd s^k<\sum_{k=0}^\infty  \p\ci L_{1}(f_{1,1},m_{1,1})=k\cd t^k=h(t),
    \end{equation}
    implying that $h$ is strictly increasing.
    
   Now, for the sequence $\{h^{(n)}(0)\}_n$ we obtain from  \eqref{cosita} and $h(1)=1$ that
    \begin{equation}\label{cosita2}
        h^{(n)}(0)<1,\quad \forall n\in \N.
    \end{equation}
    Using that $Z_0=N$ we have
    {\small
    \begin{align*}
       \p[Z_1=0]&= \p\ci L_{N}\pai\sum_{j=1}^N f_{j,0},\sum_{j=1}^N m_{j,0}\pad<1\cd\\
       &\leq \p\ci \sum_{j=1}^N L_{1}\pai f_{j,0},m_{j,0}\pad<1\cd\leq \p\ci L_{1}(f_{1,0},m_{1,0})< 1\cd^N<1,
    \end{align*}}hence, analogously to \eqref{cosita}, we deduce that
    \[
    g_1(s)<g_1(t),
    \]whenever $s<t$. In particular the latter together with \eqref{cosa} and $\eqref{cosita2}$ imply that
    \[
        g_n(0)\leq g_1(h^{(n-1)}(0))<1,\quad \forall n\in \N,
    \]which is the same as
    \[
    \p[\tau\leq n]=\p[Z_n=0]<1,\quad \forall n\in \N,
    \]and the result follows.\hfill   
\end{proof}

Condition 3 on the function $L$ and mathematical induction yield that 
$$\E\ci Z_k\cd\leq K\min\{\E\ci F_1\cd,\E\ci M_1\cd\},$$ which, since we have assumed that both sequences $\{f_{j,0}\}_j$ and $\{m_{j,0}\}_j$ have finite means, implies that  $\E\ci Z_k\cd <\infty$ for all $k\in\N$.

\begin{propo}\label{extincionligera}If $r<1$, it holds that $F_\tau\in \mathcal{H}^\complement$.
\end{propo}
\begin{proof}
For any $j\in\N$ following the arguments in page 6 from \cite{hulletal}, we have

\begin{align*}
    \p\ci \tau>k | Z_{0}=j\cd=\p\ci Z_k>0 | Z_0=j \cd\leq jr^{k}=je^{-(-\ln(r))k}
\end{align*}

hence we obtain $\overline{F}_\tau(k)\leq \E\ci Z_{0}\cd e^{-(-\ln(r))k}.$ Since $r<1$ implies $-\ln(r)>0$, setting $a=\E\ci Z_{0}\cd e^{(-\ln(r))}$ and $b=-\ln(r)$ we conclude that $\overline{F}$ satisfies Definition \ref{colaligera}.
\end{proof}

The following lemma is key for our main results.

\begin{lemma}\label{ratiocero}
    In the subcritical BGWBP, $mZ_{k+m}\to 0$ for all $k$ as $m\to\infty$.
\end{lemma}

\begin{proof}
    We use the fact that $mZ_{k+m}=mr^m\pai\frac{Z_{k+m}}{r^{k+m}}\pad r^k.$ As proved in \cite{molinaetal2}, it holds that $\frac{Z_{k+m}}{r^{k+m}}\to 0$ almost surely in the subcritical case. Since $mr^m\to 0$, the result follows.
\end{proof}

\begin{remark}
    Note that Lemma \ref{ratiocero} is true for any distribution of $m_{1,0}$ and $f_{1,0}$.
\end{remark}

Next we note that $\p\ci \tau >k+1 | Z_k\cd=\p \ci F_k>0,M_k>0 | Z_k\cd,$ where the probability is obviously zero if $Z_k=0$, and set

\begin{equation}\label{pincheH}
    \Lambda(k):=p_m^{-Z_k}+p_f^{-Z_k}-1,\quad p_f=\p\ci f_{1,0}=0\cd,p_m=\p\ci m_{1,0}=0\cd.
\end{equation}

Under extinction with probability 1, it is clear that $\Lambda(k)\overset{a.s.}{\to} 1\text{ as }k\to\infty$.

\begin{lemma}\label{lemapincheH}
    Under extinction with probability 1, the following assertions hold:
    \begin{enumerate}
        \item For all $k\in\N$ and $\beta>1$, there exists $M_0:=M_0(k)$ such that for all $M\geq M_0$,

    \[\beta M\E\ci \exp\left\{\theta(k+1)-\theta Z_{k+M}\right\}\mathbbm{1}_{\{\Lambda(k+M)\leq \beta\}}\cd>1.\]

\item For each fixed $k\in \N$, $\Lambda^M(k+M)\to 1$ almost surely as $M\to\infty$.
    \end{enumerate}
    \end{lemma}
\begin{proof}
        \begin{enumerate}
            \item If Assertion 1 is not true, there exist $k\in\N$, $\beta>1$ and a subsequence $\{M_a\}_a\subseteq \N$ such that for all $a\in \N$,
        \[\E\ci \exp\left\{\theta(k+1)-\theta Z_{k+M_a}\right\}\mathbbm{1}_{\{\Lambda(k+M_a)\leq \beta\}}\cd\leq (\beta M_a)^{-1}.\]

        If the above inequality holds, then we can take limits as $a\to\infty$ and use dominated convergence. Doing this gives $e^{\theta(k+1)}\leq 0$, which is absurd.

            \item By Lemma \ref{ratiocero},  there exists a null set $\mathcal{N}$ such that for all $\omega\in \mathcal{N}^\complement$ we have  $MZ_{k+M}\to0$ as $ M\to\infty.$

    Therefore, rewriting $p_\cdot^{-Z_k}$ as $\exp \{\alpha_\cdot Z_k\}$ where $\alpha_\cdot=-\ln(p_\cdot)$, we obtain for $\omega\in \mathcal{N}^\complement$

    \begin{align}
        \Lambda^M(k+M)(\omega)=\pai1+M\frac{e^{\alpha_f Z_{k+M}(\omega)}+e^{\alpha_m Z_{k+M}(\omega)}-2}{M}\pad^M.\label{halaM}
    \end{align}

Next we prove that $M\pai e^{\alpha_f Z_{k+M}(\omega)}+e^{\alpha_m Z_{k+M}(\omega)}-2\pad\to 0$ as $M\to\infty$. From Taylor's series representation for the exponential function, we observe that

    \begin{align*}
        &M\pai e^{\alpha_f Z_{k+M}(\omega)}+e^{\alpha_m Z_{k+M}(\omega)}-2\pad\\
        &=M\pai\sum\limits_{j\geq 1} \frac{(\alpha_fZ_{k+M}(\omega))^j}{j!}+\sum\limits_{j\geq 1} \frac{(\alpha_mZ_{k+M}(\omega))^j}{j!}\pad\\
        &=(\alpha_f+\alpha_m)MZ_{k+M}(\omega)+M\pai\sum\limits_{j\geq 2} \frac{(\alpha_fMZ_{k+M}(\omega))^j}{j!M^j}+\sum\limits_{j\geq 2}\frac{(\alpha_mMZ_{k+M}(\omega))^j}{j!M^j}\pad\\
        &\leq (\alpha_f+\alpha_m)MZ_{k+M}(\omega)+M^{-1}\pai\sum\limits_{j\geq 2} \frac{(\alpha_fMZ_{k+M}(\omega))^j}{j!}+\sum\limits_{j\geq 2}\frac{(\alpha_mMZ_{k+M}(\omega))^j}{j!}\pad\\
        &\leq (\alpha_f+\alpha_m)MZ_{k+M}(\omega)+M^{-1}\pai e^{\alpha_fMZ_{k+M}(\omega)}+e^{\alpha_mMZ_{k+M}(\omega)}\pad.
    \end{align*}

        The latter term clearly tends to zero in view that $\omega$ is such that $MZ_{k+M}(\omega)\to 0$. This and (\ref{halaM}) prove the assertion that $\Lambda^M(k+M)\to 1$ almost surely.
        \end{enumerate}
        \hfill
\end{proof}

The following two theorems are our main results.

\begin{teo}\label{teoremaimportante2}
    Assume Hypothesis \ref{hypo1} holds. Then 

\begin{equation}
\exp\left\{-\exp\left\{-\theta k\right\}\right\}\leq \liminf\limits_{m\to\infty}F_\tau^m\pai k+1+m\pad,\label{liminfbonito2}
\end{equation}

and

\begin{equation}
\limsup\limits_{m\to\infty}F_\tau^m\pai k+1+m\pad\leq \exp\left\{-\exp\left\{-\theta (k+1)\right\}\right\},\label{limsupbonito2}
\end{equation}
for all $k\geq 1$.
    
\end{teo}

\begin{proof}
    First note that
    \begin{align}
        \p\ci \tau>k+1 | Z_k\cd&=1-\p\ci \{F_k=0\}\cup\{M_k=0\}|Z_k\cd=1-p_f^{Z_k}-p_m^{Z_k}+(p_fp_m)^{Z_k}\nonumber \\
        &=1-e^{-\theta Z_k}\Lambda(k),\label{nuevacola}
    \end{align}
    where $\Lambda$ is the function defined in (\ref{pincheH}). First, we prove (\ref{liminfbonito2}), for which we take expectations in (\ref{nuevacola}). This gives

    \begin{align}
        F^M_\tau(k+1+M)&=\E\ci \exp\Big\{-\theta M Z_{k+M}+e^{-\theta k}-e^{-\theta k}\Big\}\Lambda^M(k+M)\cd\nonumber\\
        &\geq \exp\Big\{- e^{-\theta k}\Big\}\E\ci\exp\Big\{-\theta M Z_{k+M}\Big\}\Lambda^M(k+M)\cd.\label{cotainferiornuevacola}
    \end{align}

    Let $\beta>1$ be a fixed real constant. Since $\Lambda$ is non negative, splitting the expectation in the right-hand side of (\ref{cotainferiornuevacola}) in two parts yields

    \begin{align}
        F^M_\tau(k+1+M)&\geq \exp\Big\{- e^{-\theta k}\Big\}\E\ci\exp\Big\{-\theta M Z_{k+M}\Big\}\Lambda^M(k+M)\mathbbm{1}_{\{\Lambda^M(k+M)\leq \beta\}}\cd.\label{cotainferiornuevacola2}
    \end{align}
    We point out that the sequence 
    
    \[\Bigg\{\exp\Big\{-\theta M Z_{k+M}\Big\}\Lambda^M(k+M)\mathbbm{1}_{\{\Lambda^M(k+M)\leq \beta\}}\Bigg\}_M,\]
    is non negative and it is bounded by $\beta$ and that $\mathbbm{1}_{\{\Lambda^M(k+M)\leq \beta\}}$ tends to 1 thanks to Lemma \ref{lemapincheH}, part 2.
    Therefore, taking limit inferior, applying dominated convergence in the right-hand side of (\ref{cotainferiornuevacola2}) and using Lemma \ref{lemapincheH} part 2, we obtain

    \[\liminf\limits_{M\to\infty}F^M_\tau(k+1+M)\geq \exp\Big\{- e^{-\theta k}\Big\}.\]

    To prove (\ref{limsupbonito2}) we must first find a convenient lower bound for

    \begin{align*}
        M\overline{F}_\tau(k+1+M)&=M\E\ci1-\exp\left\{-\theta Z_{k+M}\right\}\Lambda(k+M) \cd.
    \end{align*}

Here we use part 5 of Hypothesis 1, which implies that $(2,e^{\theta})$ is a non empty interval. Hence we take $\beta\in (1,e^{\theta}/2)$ and split $\E\ci1-\exp\left\{-\theta Z_{k+M}\right\}\Lambda(k+M) \cd$ as before. This yields

        \begin{align}
        M\overline{F}_\tau(k+1+M)&\geq M\E\ci\pai 1-\exp\left\{-\theta Z_{k+M}\right\}\Lambda(k+M)\pad\mathbbm{1}_{\{\Lambda(k+M)\leq \beta\}} \cd\nonumber \\
        &\geq M\E\ci\pai 1-\exp\left\{-\theta Z_{k+M}\right\}\beta\pad\mathbbm{1}_{\{\Lambda(k+M)\leq \beta\}} \cd\nonumber \\
        &\ =M\E\ci\exp\Big\{\ln\pai 1-\beta \exp\left\{-\theta Z_{k+M}\right\}\pad\Big\}\mathbbm{1}_{\{\Lambda(k+M)\leq \beta\}} \cd.\label{cotainferiorcolanueva3}
    \end{align}
By the choice of $\beta$, it is true that $\ln(2\beta)<\theta$, hence $1-\beta \exp\left\{-\theta Z_{k+M}\right\}\geq \beta \exp\left\{-\theta Z_{k+M}\right\} $ and (\ref{cotainferiorcolanueva3}) becomes

        \begin{align}
        M\overline{F}_\tau(k+1+M)&\geq M\E\ci\exp\Big\{\ln (\beta)-\theta Z_{k+M}\Big\}\mathbbm{1}_{\{\Lambda(k+M)\leq \beta\}} \cd\nonumber\\
        &\ =\beta M\E\ci\exp\Big\{ -\theta Z_{k+M}\Big\}\mathbbm{1}_{\{\Lambda(k+M)\leq \beta\}} \cd\nonumber \\
        &=\exp\left\{ -\theta(k+1)\right\}\Bigg(\beta M\E\ci \exp\left\{\theta(k+1)-\theta Z_{k+M}\right\}\mathbbm{1}_{\{\Lambda(k+M)\leq \beta\}} \cd\Bigg).\label{cotainferiorcolanueva4}
    \end{align}

    We're ready to prove (\ref{limsupbonito2}). We use the inequality $1-x\leq e^{-x}$ for $x\in [0,1]$ together with (\ref{cotainferiorcolanueva4}), which yields

    \begin{align*}
        F_\tau^M(k+1+M)&\leq \exp\Big\{-M\overline{F}_\tau(k+1+M)\Big\}\\
        &\leq \exp\Bigg\{-e^{-\theta (k+1)}\pai\beta M\E\ci \exp\left\{\theta(k+1)-\theta Z_{k+M}\right\}\mathbbm{1}_{\{\Lambda(k+M)\leq \beta\}} \cd\pad\Bigg\}\\
        &\leq exp\Big\{-e^{-\theta (k+1)}\Big\},
    \end{align*}
    where the last inequality holds for a sufficiently large $M$ (by Lemma \ref{lemapincheH}, part 1). Taking limit superior in the inequality above gives (\ref{limsupbonito2}).
\end{proof}

 \begin{teo}\label{teoremabonito2}
    Under the hypothesis of Theorem \ref{teoremaimportante2}, the extinction time $\tau$ is the discretization of a random variable $X$ with distribution function $F_X$, such that $F_X\in D(H_\text{Gumbel})$.
\end{teo}

\begin{proof}
    By Theorem \ref{teoremaimportante2} and Theorem 2 in \cite{anderson}, we obtain 
    $$\frac{\overline{F}_\tau(n+1)}{\overline{F}_\tau(n)}\to e^{-\theta},\quad \text{as }n\to\infty.$$
    Since $\theta>0$ implies $e^{-\theta}\in(0,1)$, Corollary 1 in \cite{shimura} gives the result.
\end{proof}

\begin{remark}Theorem 2 in \cite{anderson} requires the existence of constants $\{b_n\}_n$, which in our case are given by $b_n=n$. These constants are not unique. They can be replaced by any sequence $\{b'_n\}_n$ such that $nr^{b'_n}\to 0$, so that Lemma \ref{ratiocero} remains true. For instance, if we take $b'_n=n^\alpha$ for $\alpha>1$, we obtain using L'Hôpital's rule and the hypothesis $r<1$,

\[ \lim\limits_{x\to\infty}\frac{x}{e^{-x^\alpha\ln(r)}}=\lim\limits_{x\to\infty}\frac{1}{-\alpha x^{\alpha-1}\ln(r)e^{-x^\alpha\ln(r)}}=0.\]
    
\end{remark}

\section{Numerical Examples}

By Theorem \ref{teoremabonito2} we know that $F_\tau$ can be approximated using EVT. In this section we present some numerical examples of how this approximation is obtained using data. All our data were  simulated using specific mating functions and offspring distribution.

Before we explain the methodology, we consider the following result.

\begin{lemma}\label{lemaextincion}
    If $L(x,y)\leq Cx$ and $\E\ci f_{1,0}\cd<C^{-1}$ or if $\E\ci m_{1,0}\cd<1$ and $L_n(x,y)\leq y$ for all $n$, then $r<1$.
\end{lemma}

\begin{proof}
Without loss of generality, let us assume $L_n(x,y)\leq Cx$ and $\E\ci f_{1,0}\cd<C^{-1}.$

We only need to note that for any $k\in \N$, by the third condition on the mating functions and independence of $m_{j,0},f_{j,0}$ for each $j$, we have
\[
r_k=\frac 1k\E\ci L\pai\sum_{j=1}^k f_{j,0},\sum_{j=1}^k m_{j,0}\pad\cd\leq \frac Ck\E\ci\sum_{j=1}^k f_{j,0}\cd=C\E[f_{1,0}]<1.
\]
\hfill
\end{proof}

\subsection{Methodology}

In this section, we present six numerical examples with simulated values of $\tau$ under different conditions. It can be easily checked that the chosen parameters are such that the conditions of Lemma \ref{lemaextincion} are fulfilled, therefore we have extinction with probability 1.

Our approximation relies on the POT method, for which we recall Pickands, Balkema and de Haan's theorem (also known as The Second Extreme Value Theorem). In order to state this result, we recall that given random variable $X$ with distribution function $F$, the random variable $X-u$ conditioned on $X>u$ is the \emph{excess of $X$ over $u$}. 

When $F(0)=0$, the excess of a random variable over a given threshold $u$ has the distribution function

\[F_u(x) = \p[X - u \leq  x| X > u] = \frac{F(x+u)-F(u)}{\overline{F}(u)}\mathbbm{1}_{\{0\leq x \leq \omega_F -u\}}.\]

\begin{teo}(Pickands-Balkema-de
Haan Theorem) For some positive and measurable function $a(u)$, it holds that 
\begin{equation}
    \lim_{u\to\omega_F}\sup_{0\leq x \leq\omega_F-u}|F_u(x)-P_{\xi,a(u),0}|=0,\label{limitepickands}
\end{equation}
if and only if $F \in D(H_\xi)$ where $H_\xi$ is an extreme value distribution with shape parameter $\xi$ and $P_{\xi,a,b}$ is the Generalized Pareto distribution given by 
\[P_{\xi,a,b}(x) = \begin{cases}
    1-[1+\xi\left(\frac{x-b}{a}\right)]^{-\frac{1}{\xi}},\quad \text{if }\xi \neq 0,\\
    1-\exp\left\{-\frac{x-b}{a}\right\},\quad \text{if } \xi = 0.
\end{cases}\]
The expression above holds for $x\geq b$ when $\xi \geq 0$ and for $b\leq x\leq -a/\xi+b $ when $\xi<0$. In any other case, $P_{\xi,a,b}(x)=0$.
\end{teo}

\begin{remark}\label{remarkexpo}
    Note that the GPD reduces to a truncated exponential distribution when $\xi=0$, which is the case when $F\in D(H_\text{Gumbel})$.
\end{remark}
\subsection{Examples}

In view of this last result, the approximation we will perform using the result in Theorem \ref{teoremabonito2} is the one given by (\ref{limitepickands}). Since our approximation depends on discrete data, we need the following result.

\begin{propo}
    Let $Y\geq 0$ be a continuous random variable with distribution function $F_Y$ and let $X=\lfloor Y\rfloor$ be its discrete version. It holds that $\p\ci X>k\cd=1-F_Y(k)$ for all $ k\in\N.$
\end{propo}

\begin{proof}
    Since $P\ci X=k\cd=F_Y(k)-F_Y(k-1)$, we have

    $$\p\ci X\leq k\cd=\sum\limits_{j=1}^k\p\ci X=j\cd=F_Y(k)-F_Y(0)=F_Y(k).$$
    \end{proof}

This last proposition implies that when $F_Y\in D(H)$ we may estimate the shape parameter $\xi$ using data from the discrete random variable $X$. The technical details of such approximation are standard in EVT, and therefore, we omit them.

All the examples here are obtained using R. The GPD fit in each case is performed using the library ISMEV, available in R. 
\begin{remark}
    The methodology for choosing the threshold $u$ relies on standard EVT tools, which are covered in the wide related literature (see e.g. \cite{embrechtsetal} and the references therein). Therefore, such methodology is not discussed in this work.
\end{remark}

In what follows, we denote

$$f_1:=f_{1,0},m_1:=m_{1,0}.$$

\subsubsection{Case when \texorpdfstring{$f_1,m_1$}{fm} are iid binomially distributed}
For the first couple of examples we consider the following mating function independent of $n$,
\begin{equation*}
    L(x,y) = \min\{x,y\},
\end{equation*}
along with the following two cases:

\begin{enumerate}
    \item [\text{Case A.}] $m_1,f_1$ both follow a binomial distribution $Bin(10,0.049)$ and $Z_0=100$.
    \item [\text{Case B.}] $m_1,f_1$ follow a binomial distribution $Bin(5,0.19)$ and $Z_0=500$.
\end{enumerate}

It can be easily checked that all the conditions in Hypothesis \ref{hypo1} are fulfilled.

We simulated 50 000 trajectories of the BGWBP process and captured the time to extinction, which we also refer to as \emph{extinction generation}. The scatter plots for each case are presented below.
\begin{figure}[H]
	\centering
 \subfigure[]
	{\includegraphics[scale=.14]{ 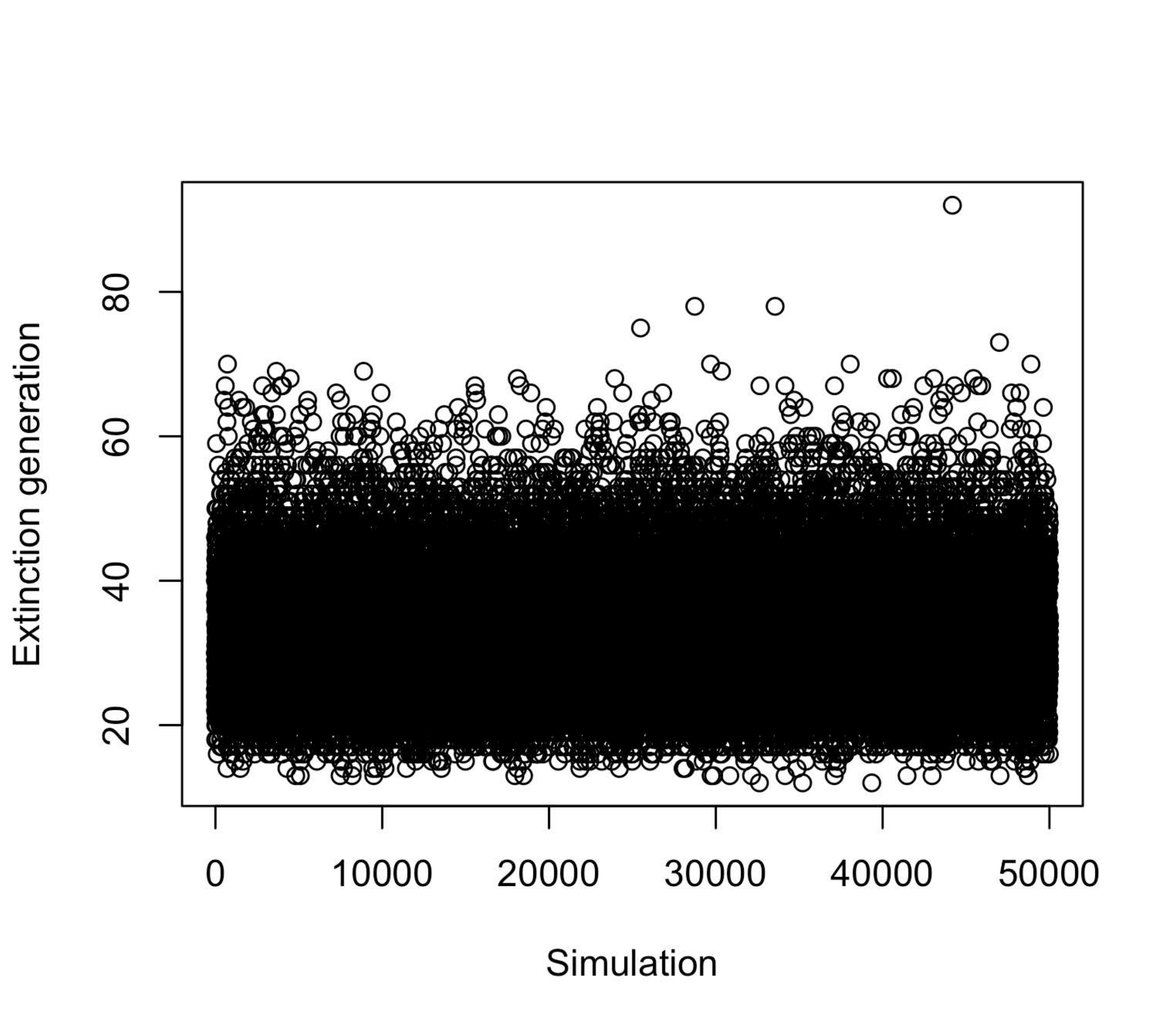}}
	\quad
	\subfigure[]
    {\includegraphics[scale=.14]{ 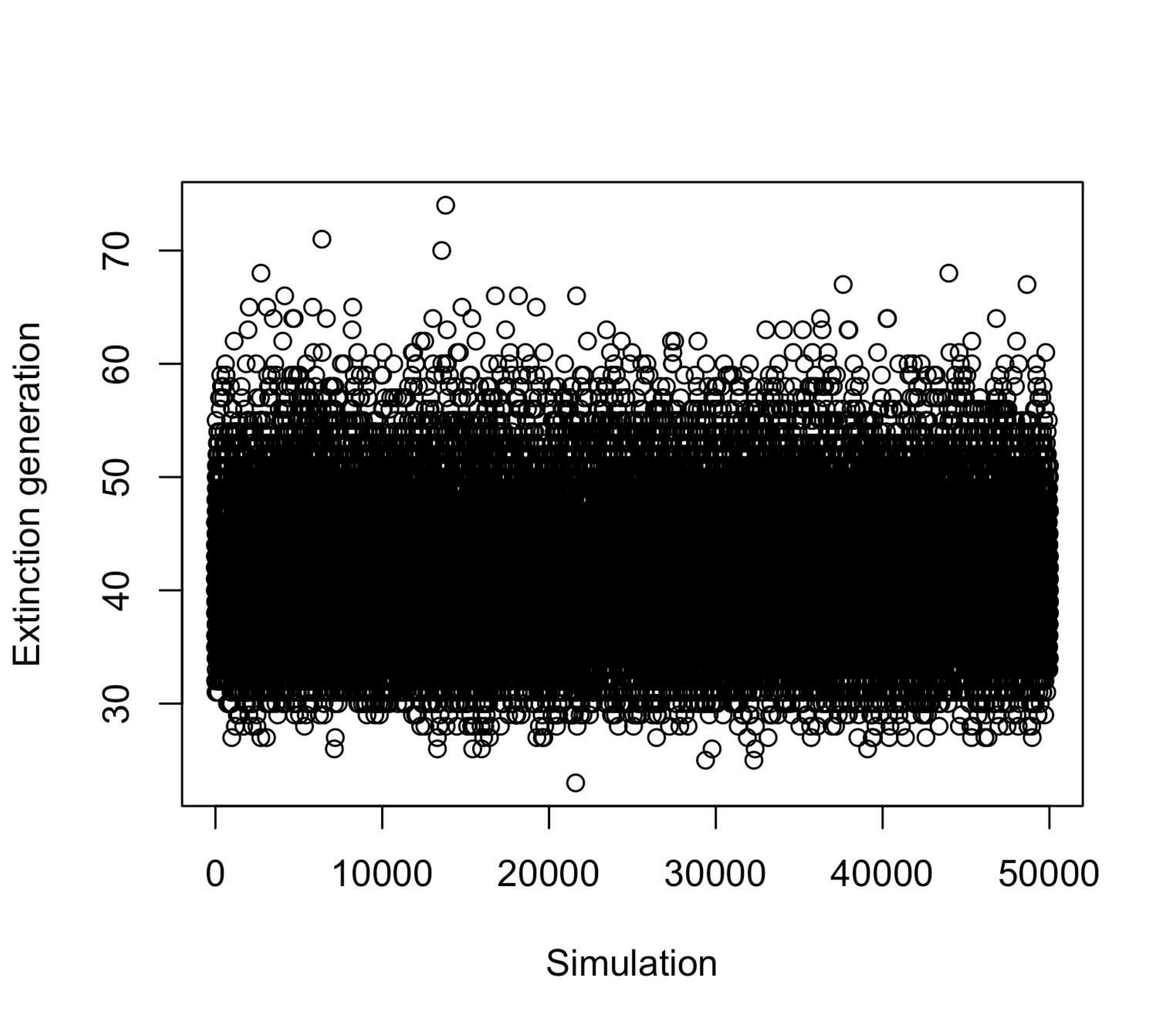}}
	{\footnotesize\caption{Simulated values of $\tau$ for case A (a) and case B (b). The fact that not all data lays in the same rectangular region suggests that the right endpoint of this distribution is not finite, which coincides with Proposition \ref{omegainfinito}.\\
  Furthermore, there are very few values that drastically differ from the others, which suggests that this data set comes from a light-tailed distribution (in agreement with Proposition \ref{extincionligera}).\\
 The data corresponding to (a) has an estimated mean of 32.29416 with an estimated standard deviation of 7.791591, while data for (b) have estimated mean and standard deviation given respectively by 41.46776 and 5.274994.}}
	\label{fig:mesh3}
\end{figure}

Using standard tools from EVT, we obtain the following estimated parameters and values of $u$:

\begin{itemize}
    \item For the case of the binomial distribution $Bin(10,0.049)$,
    \[u=45, \quad \hat{\xi} = -0.0893, \quad \hat{a}(45) = 5.85.\]
    \item For the case of the binomial distribution $Bin(5,0.19)$,
    \[u = 47.5, \quad \hat{\xi} = -0.083,\quad \hat{a}(47.5) = 3.56.\]
\end{itemize}

In both cases the estimated value of $\xi$ is very close to zero, as expected. Now we look at the Q-Q plots to check goodness of fit. These plots are presented in the figure below.

\begin{figure}[H]
	\centering
	\subfigure[]
	{\includegraphics[scale=.14]{ 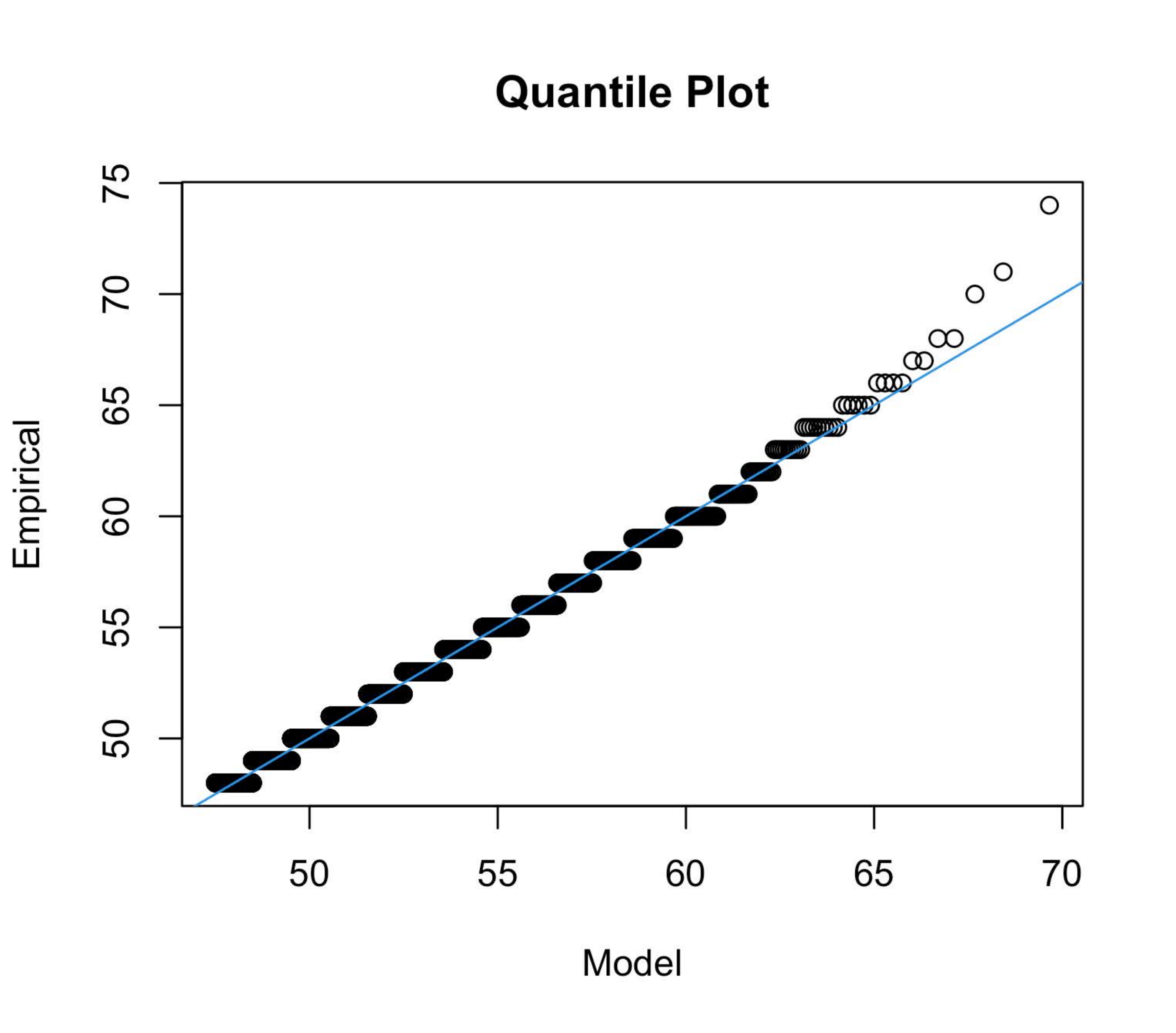}}
    \subfigure[]
	{\includegraphics[scale=.14]{ 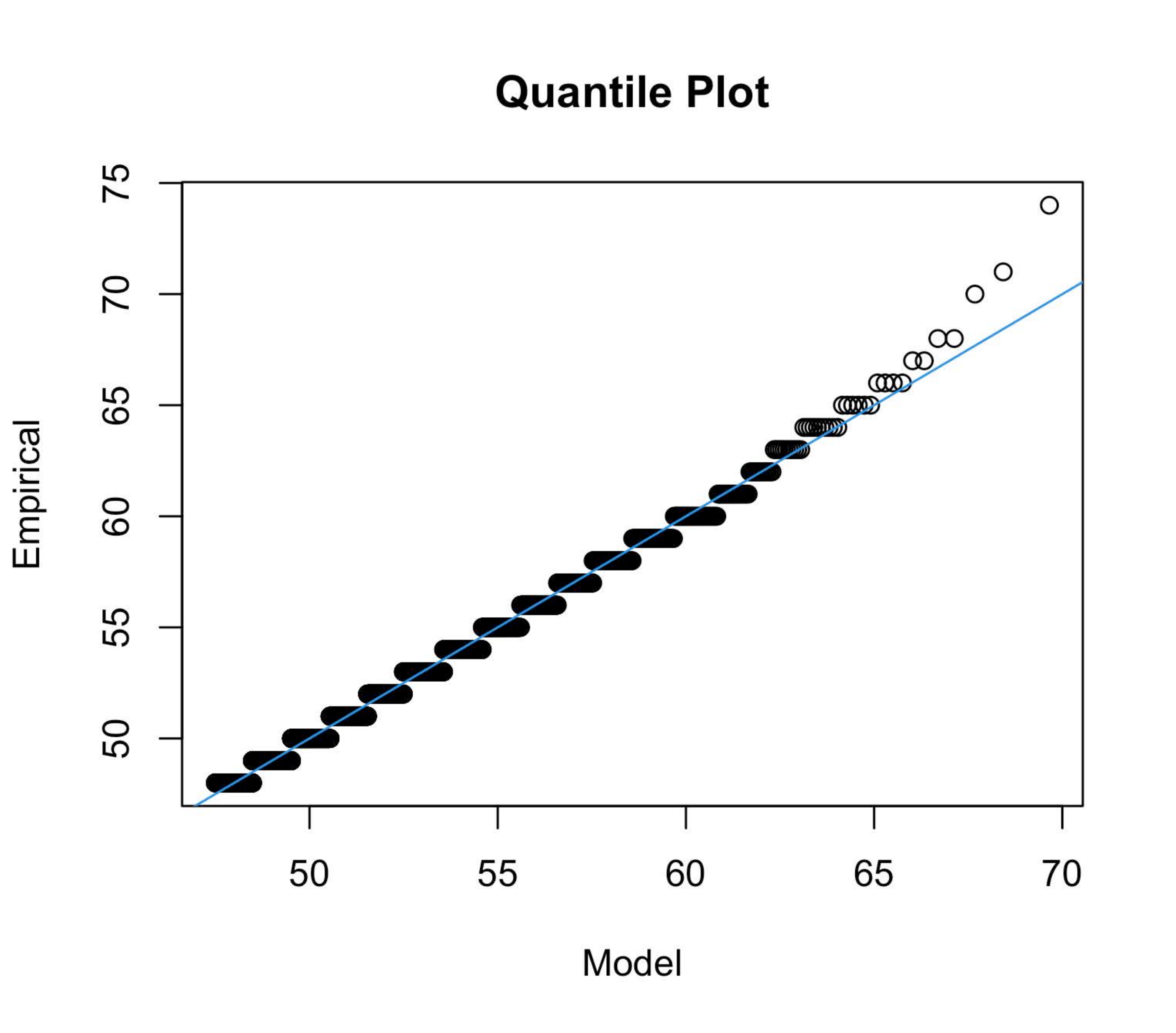}}
	{\footnotesize \caption{Goodness of fit of a GPD for case A (a) and case B (b). The Q-Q plots show that the behavior of the data indeed looks like a discretization of a data set to which the GPD would fit well.}}
	\label{fig:mesh1}
\end{figure}
Since the estimated values of the shape parameter $\xi$ in each case are close to zero, we expect that fitting an exponential distribution to the excesses would produce a Q-Q plot that behaves at least as good as those in the GPD fit. We present such plots below.
\begin{figure}[H]
	\centering
 \subfigure[]{ \includegraphics[scale=.14]{ 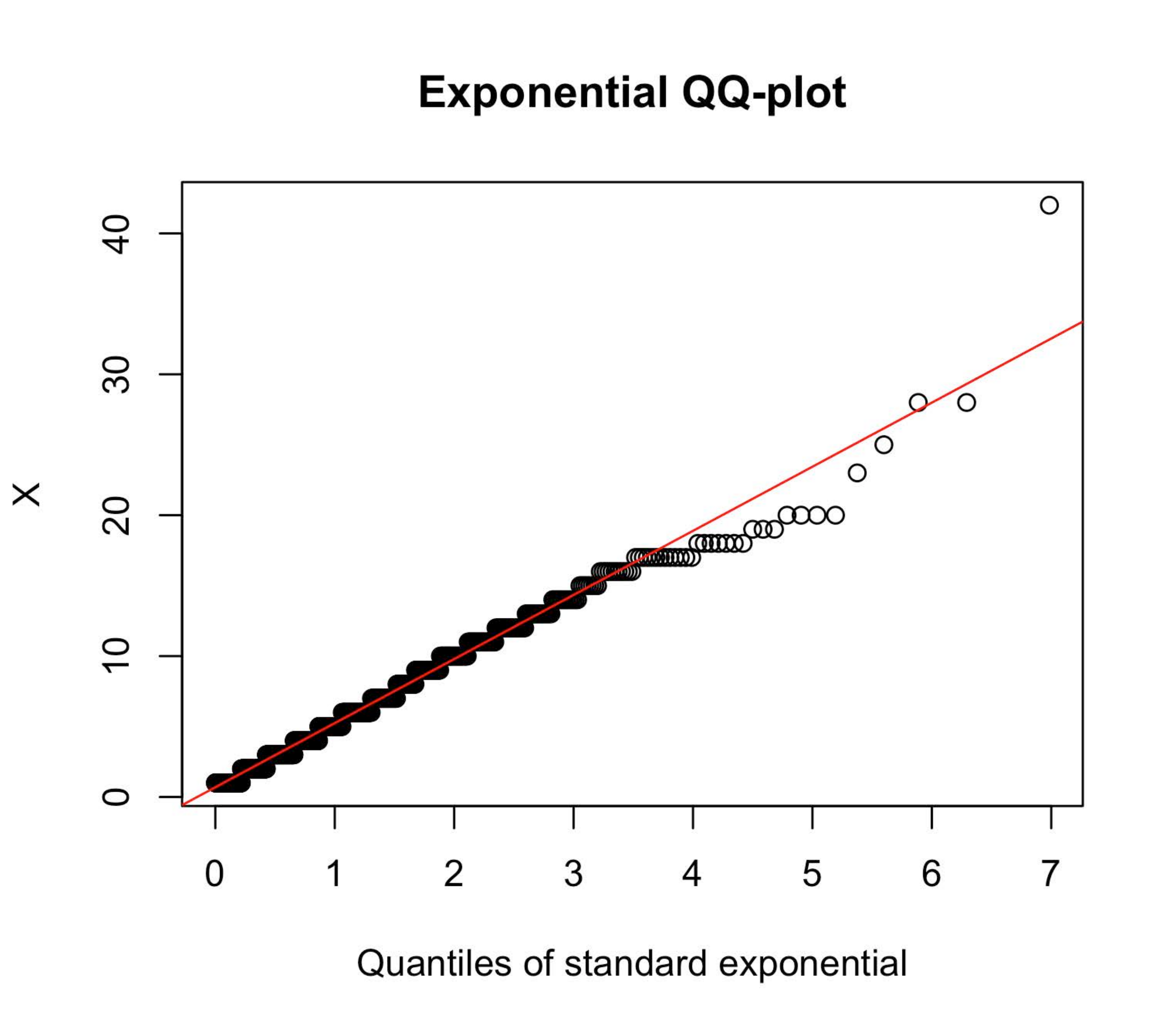}}\quad \quad
 \subfigure[]{ \includegraphics[scale=.14]{ 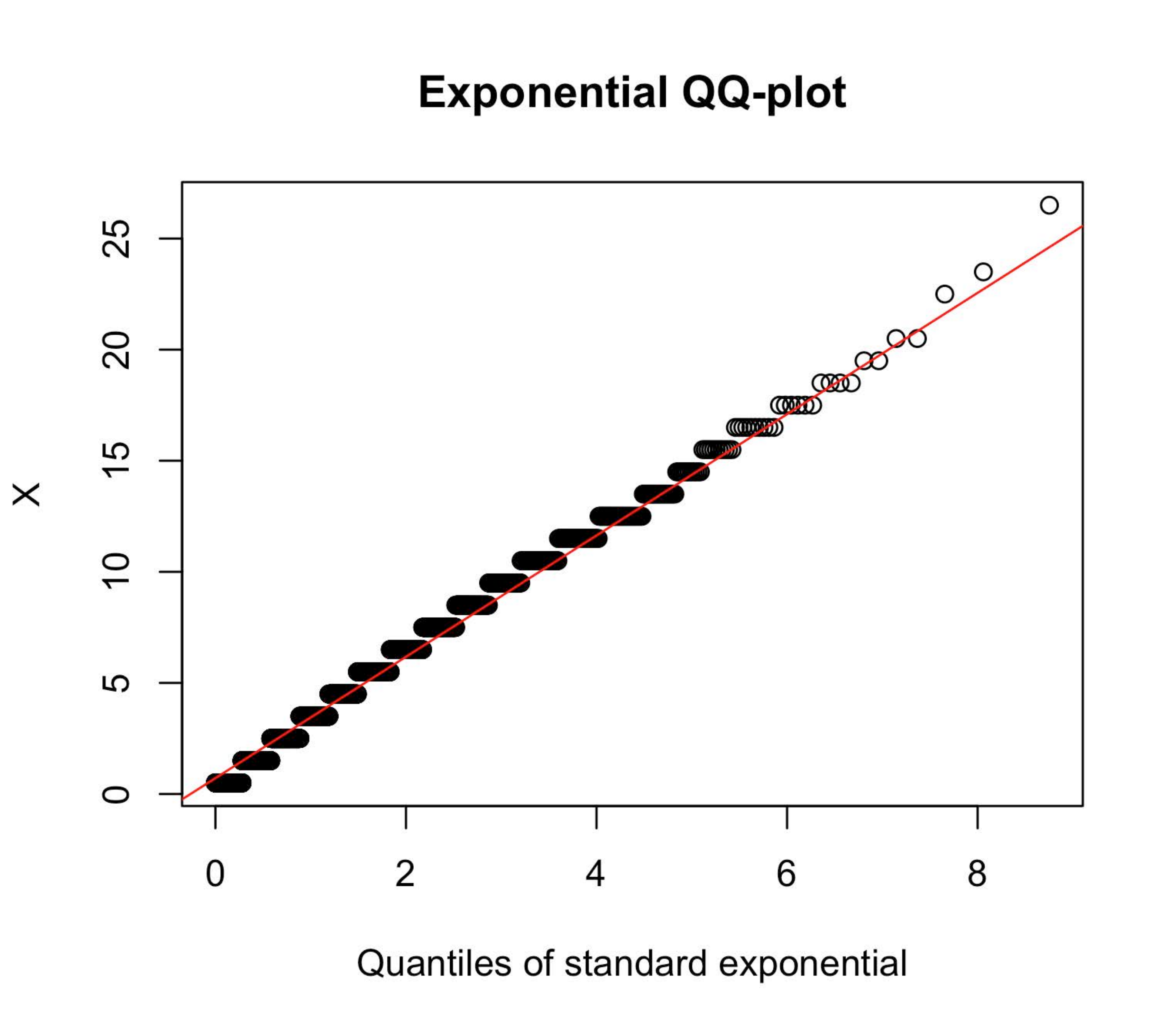}}
	{\footnotesize \caption{Goodness of fit of an exponential distribution to the data under case A (a) and case B (b). Again we see the ladder-like behavior of the data, which indicates a nice fitting of this distribution to the given data set.}}
	\label{fig:mesh}
\end{figure}

\subsubsection{Cases with a mating function depending on the previous generation}

Now we assume that $f_1\sim Geometric(0.6)$ supported on $\N\cup\{0\}$ and $m_1\sim Poisson(1.4)$. The mating functions now depend on a parameter related to the previous generation.
\begin{enumerate}
    \item [\text{Case }A.] $L_{Z_n}(F_n,M_n)=\min\{F_n,nM_n\}$,
    \item [\text{Case }B.] $L_{Z_n}(F_n,M_n)=\min\{F_n,Z_nM_n\}$.
\end{enumerate}

This time we simulate 50 00 trajectories of the BGWBP process with an initial number of mates equal to 100.
\begin{figure}[H]
	\centering
 \subfigure[]
	{\includegraphics[scale=.14]{ 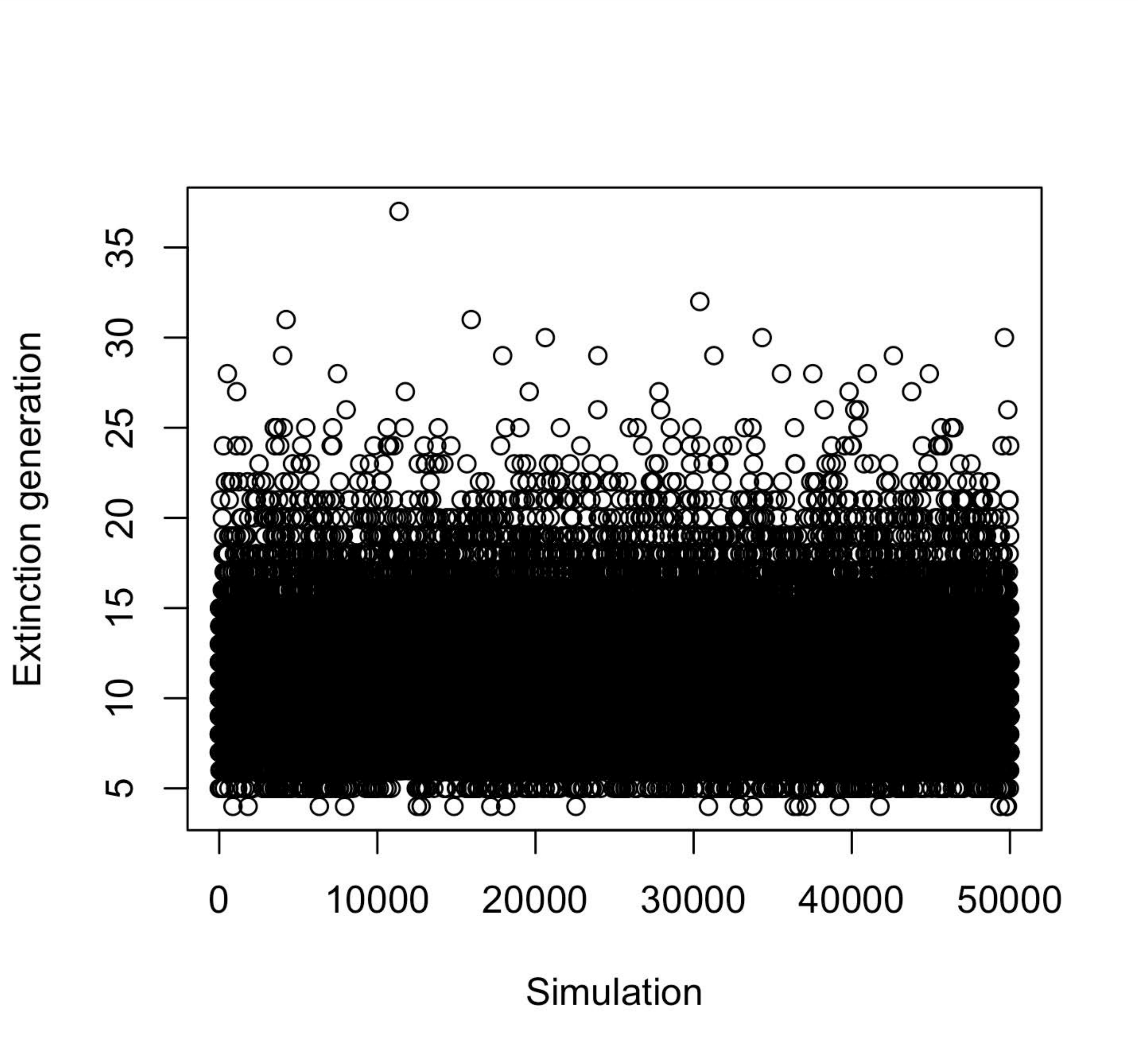}}
	\quad
	\subfigure[]
    {\includegraphics[scale=.14]{ 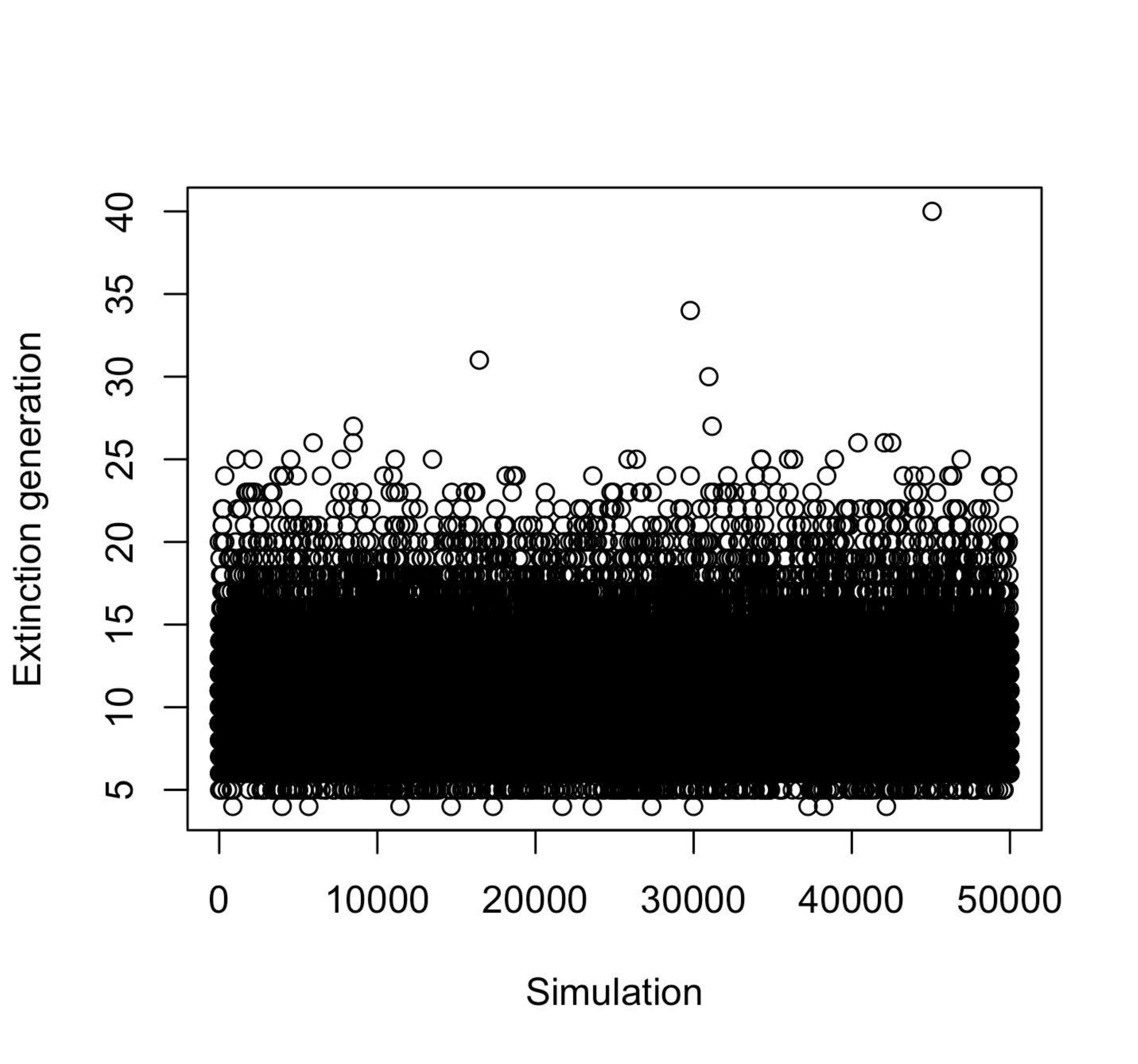}}
	{\footnotesize\caption{Simulated values of $\tau$ under case A (a) and case B (b) 
 The data corresponding to (a) has an estimated mean of 4.39602 with an estimated standard deviation of 1.883328, while data for (b) have estimated mean and standard deviation given respectively by 4.29541 and 1.684622. In both cases, the dispersion of the data seems to support the fact that the distribution behind each data set has an infinite endpoint and a light-tail.}}
	\label{fig:scatterE34}
\end{figure}

The fitting of a GPD is now performed. The values of $u$ along with the estimated parameters in these two examples are

\begin{enumerate}[\text{Case }A.]
    \item $u=11.56,\hat{\xi}=-0.009484283, \hat{a}(11.56)=2.314129977.$
    \item$u=7.46,\hat{\xi}=-0.005658594, \hat{a}(7.48)=2.025406077.$
\end{enumerate}

Note that again, both estimations of $\xi$ are very close to zero and the estimated values of $a(u)$ do not differ much.

\begin{figure}[H]
	\centering
	\subfigure[]
    {\includegraphics[scale=0.14]{ 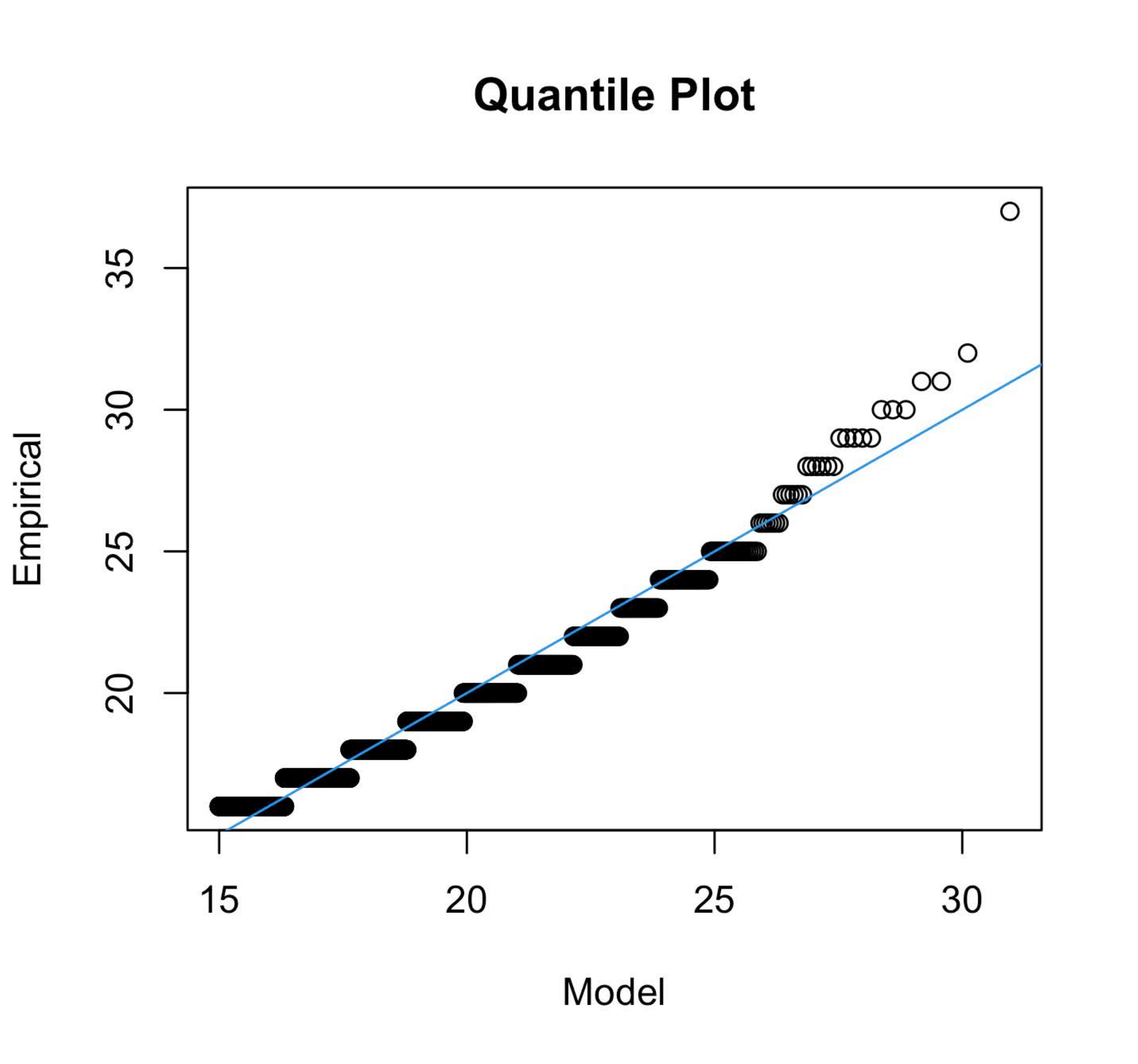}}
    \subfigure[]
    {\includegraphics[scale=0.14]{ 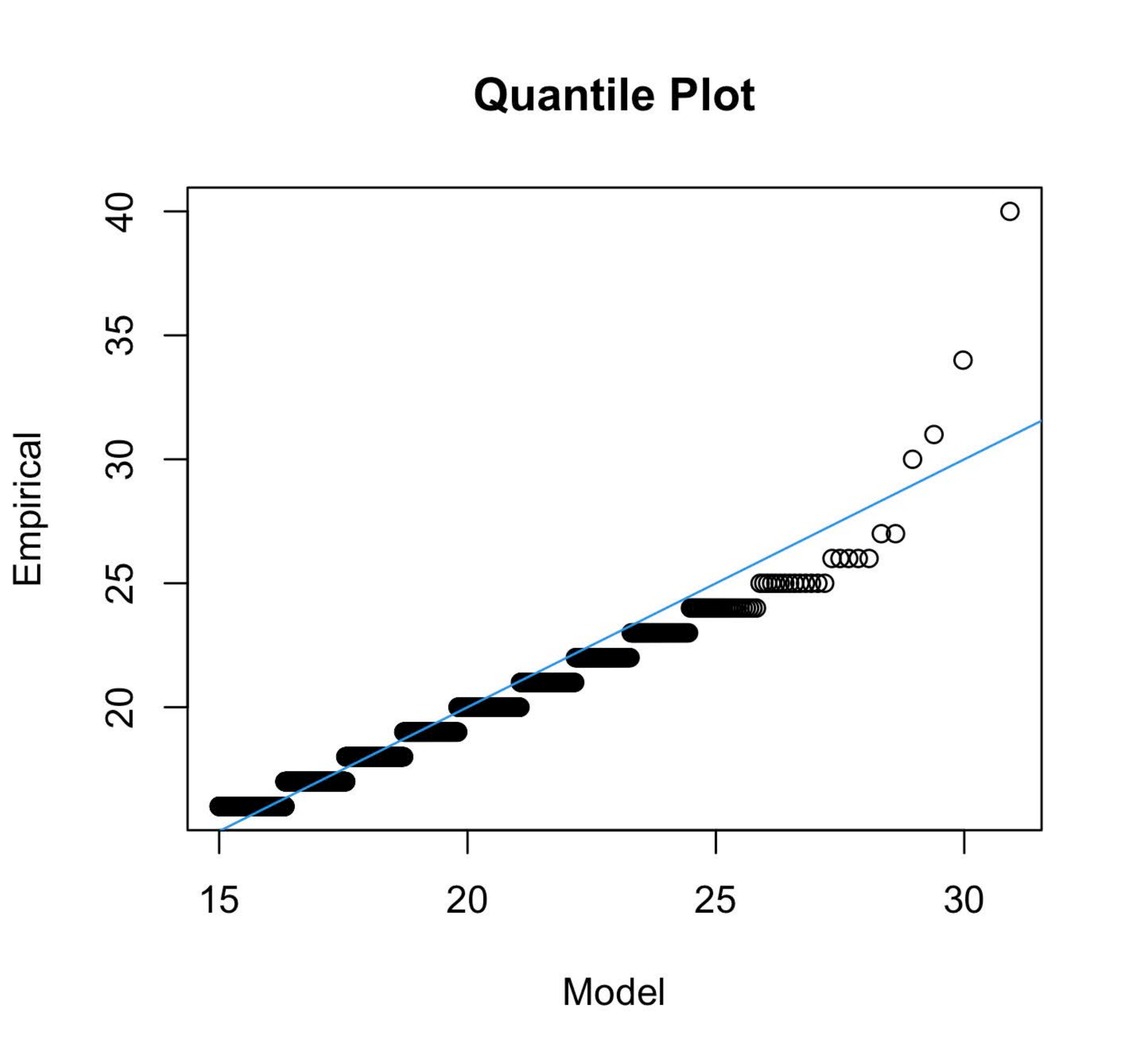}}
	{\footnotesize \caption{Goodness of fit of a GPD for case A (a) and case B (b). The Q-Q plots support the fac that data come from a discretizatión of a distribution in some maximal domain of attraction.}}
	\label{fig:GOFEx34}
\end{figure}

Since the estimated values of the shape parameter $\xi$ are very close to zero, this gives evidence that the data come from a distribution which is the discretization of some other distribution in the Gumbel maximal domain of attraction, which is the case when $\xi=0$. In view of this we explore the fitting of an exponential distribution to the data. Below we present the Q-Q plots of this exponential fitting. 
\begin{figure}[H]
	\centering
 \subfigure[]{ \includegraphics[scale=.14]{ 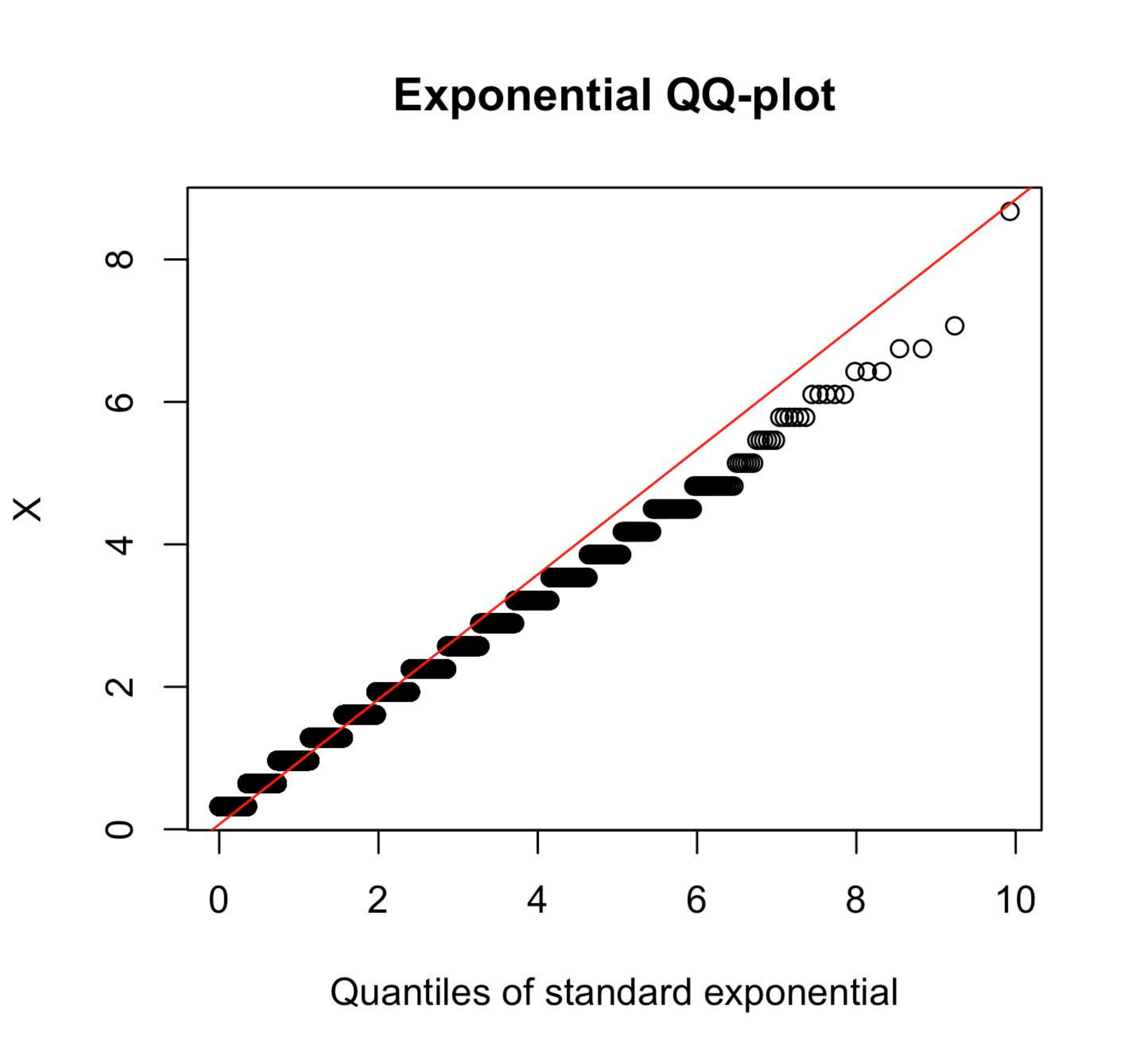}}\quad 
 \subfigure[]{ \includegraphics[scale=.14]{ 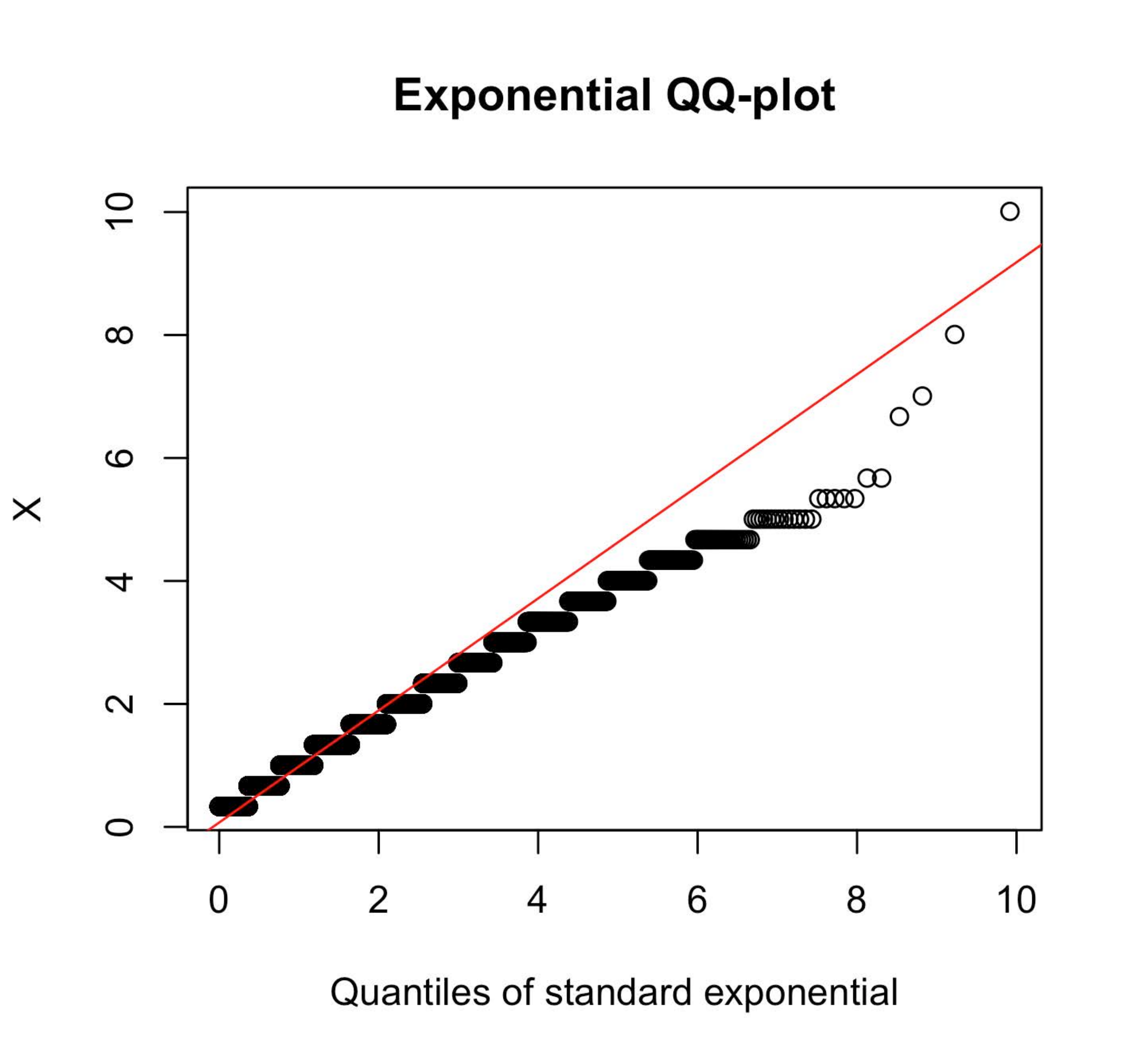}}
	{\footnotesize \caption{Goodness of fit of an exponential distribution to the data from case A (a) and case B (b). They do not look too different from those obtained in the GPD fit.}}
	\label{fig:QQExpEx34}
\end{figure}

\subsubsection{Cases with larger variance and wider range}
In these last two examples we consider

\begin{enumerate}
    \item [\textbf{Case A.}]$L_{Z_n}(F_n,M_n)=F_n\min(2,M_n)$, $f_1,m_1\sim Bin(10,0.049)$,
    \item [\text{Case B.}] $L_{Z_n}(F_n,M_n)=\min\{F_n,Z_nM_n\}$, $f_1\sim Geometric(0.51)$ supported on $\N\cup\{0\}$ and $m_1\sim Poisson(3)$.
\end{enumerate}

In both cases the initial population is set to 100 and the total number of simulations is 50 000. In the first case we have $\theta=1.004824$ while in the second case $\theta=3.6733446$.

The mating function in (b) represents the case when there are at most as many mates as females in the current generation, so each female forms at most one mate. We see how the choice of this mating function changes the adjustment of the associated GPD.

We begin by presenting the scatter plots of the simulated data. We clearly see that the range of $\tau$ has increased with respect to the previous four examples.

\begin{figure}[H]
    \centering
    \subfigure[]{\includegraphics[scale=0.14]{ 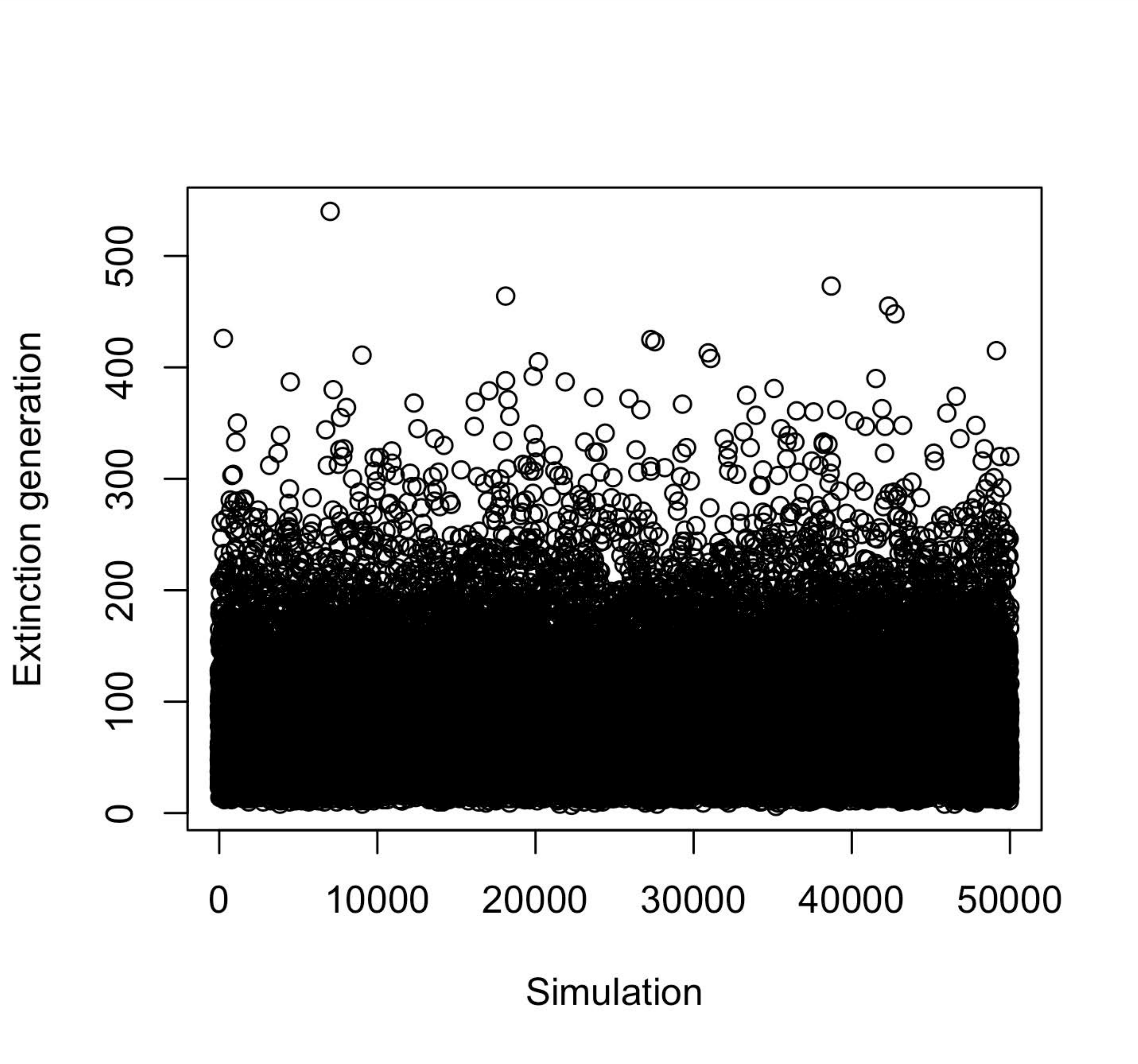}} \subfigure[]{\includegraphics[scale=0.14]{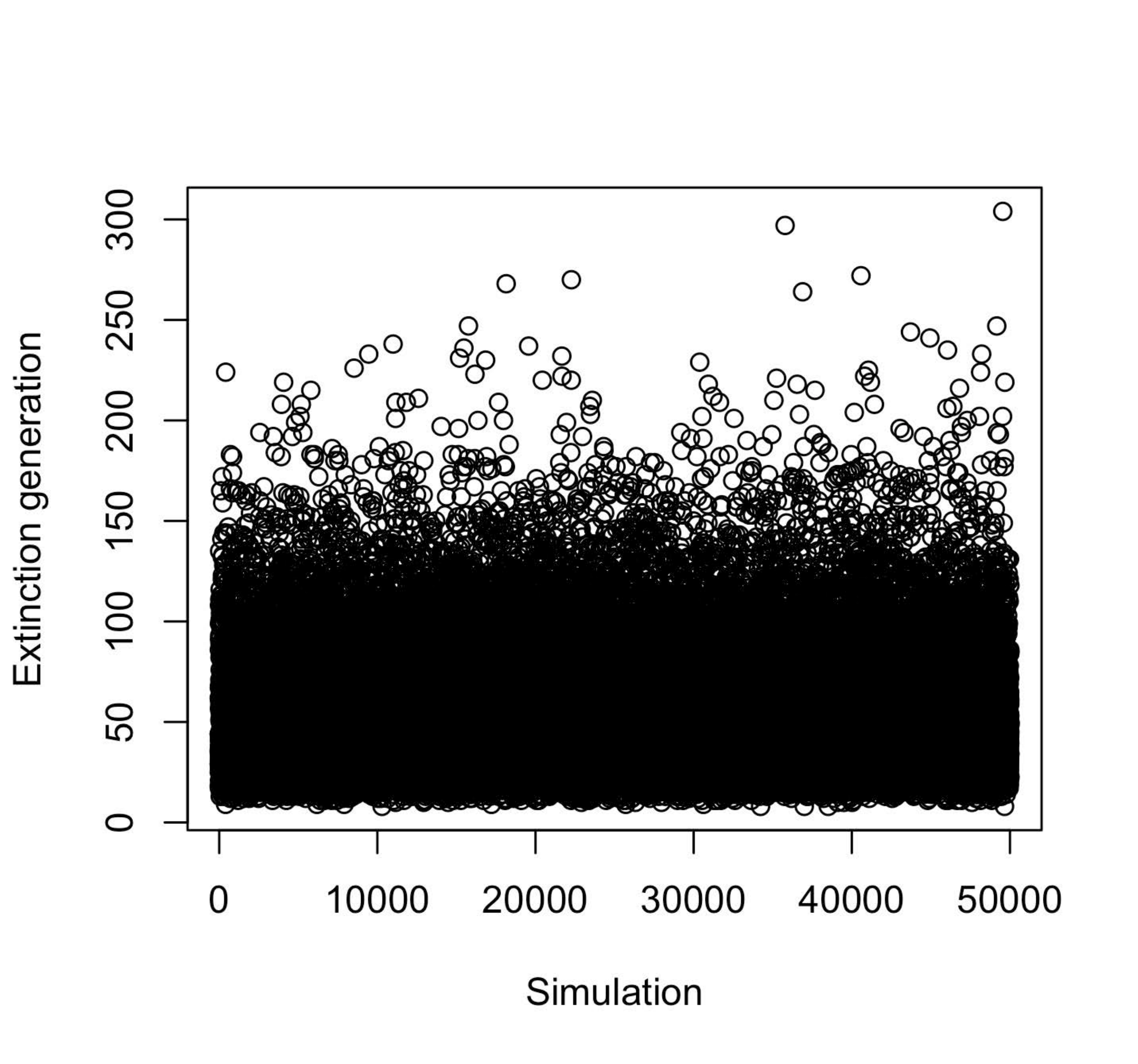}}
    \caption{{\footnotesize Scatter plots for $\tau$ under case A (a) and case B (b). The mating function in (a) gives a wider range for $\tau$, which is expected given that the estimated variance and standard deviation for the data are, respectively, 2309.401 and 48.05622, with an estimated mean of 74.63393. For the mating function in (b), the variance is 809.244 and the standard deviation is 28.4472 (both estimated using the data). The estimated mean in this case equals 56.0819.}}
    \label{fig:scatterplotsEx4and5}
\end{figure}

Now we present the fitting of a GPD to the excesses $X-u$ given $X>u$. The chosen values of $u$ and the estimated parameters are

\begin{itemize}
    \item For $L_{Z_n}(F_n,M_n)=F_n\min\{2,M_n\}$,
    \[u = 160,\quad \hat{\xi}= -0.00678,\quad \hat{a}(160)=44.38284.\]
    \item For $L_{Z_n}(F_n,M_n)=\min\{F_n,Z_nM_n\}$,
    \[u= 105,\quad \hat{\xi}= -0.03488,\quad \hat{a}(105)=26.34423.\]
\end{itemize}

Note that in both cases the shape parameter $\hat{\xi}$ is very close to zero and, as in the two previous examples, the scale parameters $\hat{a}(u)$ are very much alike. The Q-Q plots to check goodness of fit are given below.

\begin{figure}[H]
    \centering
    \subfigure[]{\includegraphics[scale=0.14]{ 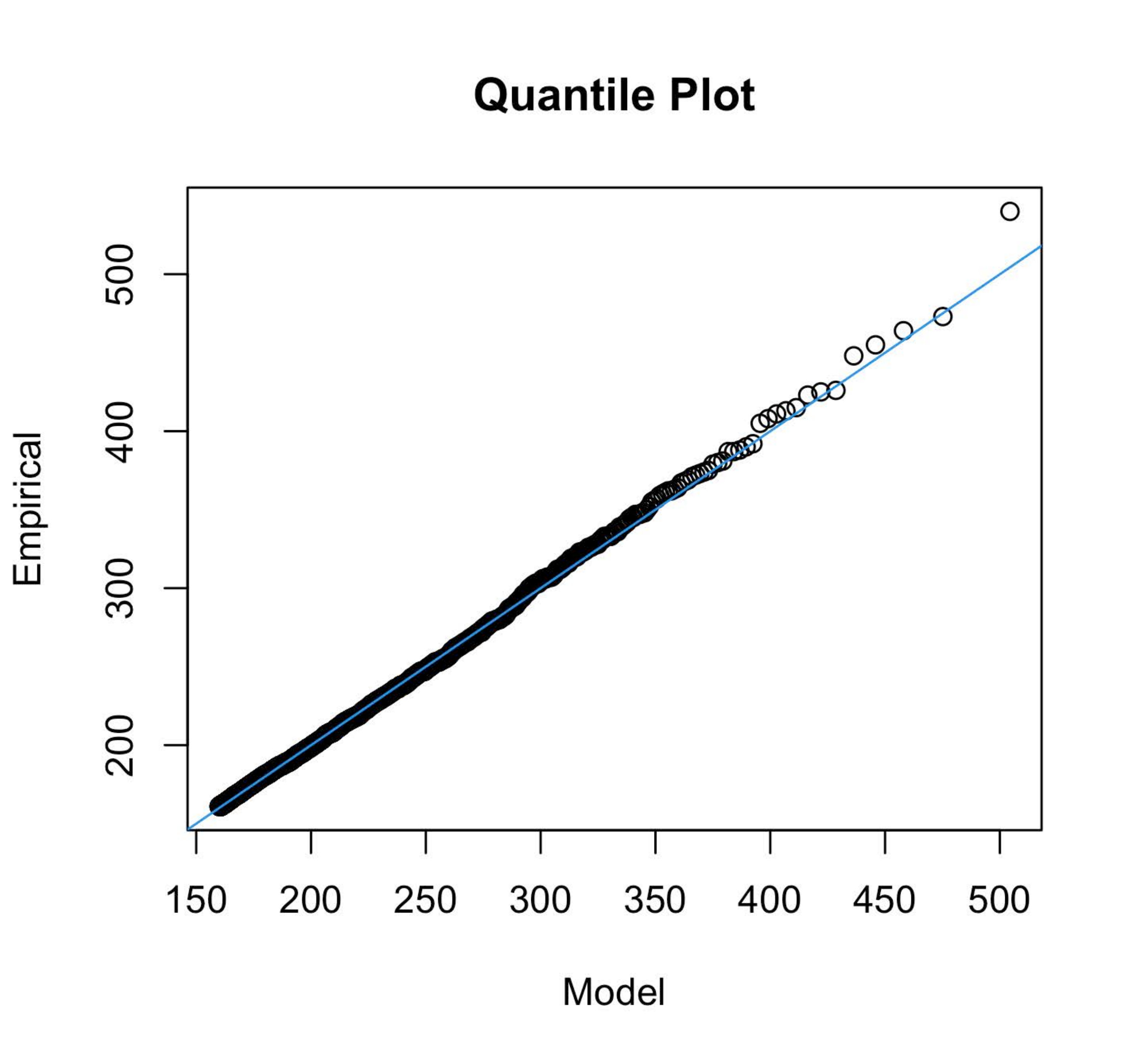}} \subfigure[]{\includegraphics[scale=0.14]{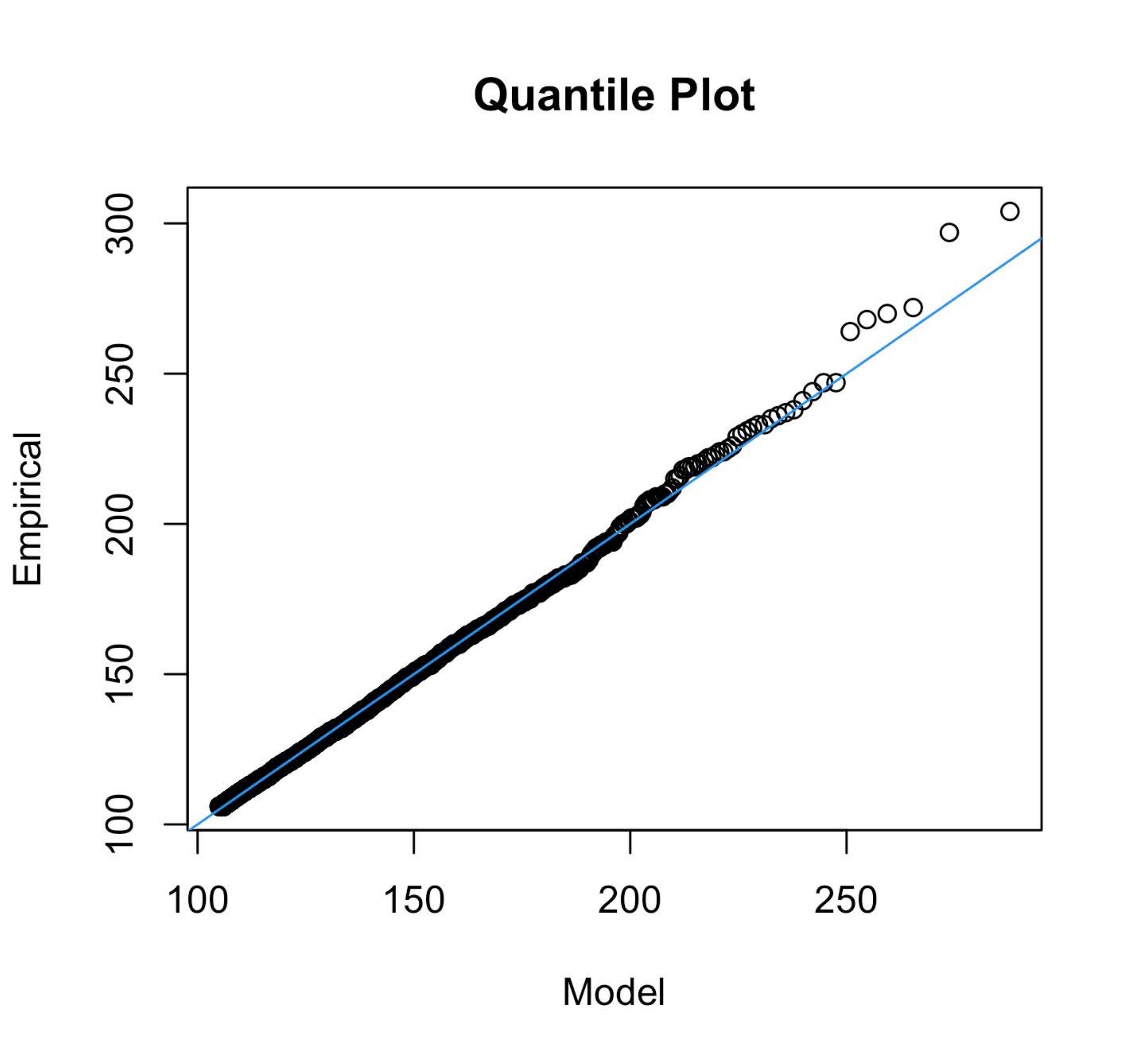}}
    \caption{{\footnotesize Q-Q plots for the GPD fitted to the excesses of $\tau$, under case A (a) and Case B (b). Unlike the previous examples, the plots show a very good fit to the corresponding identity lines, a behavior that might be due to the total number of data and the higher variances.}}
    \label{fig:QQPPEx4and5}
\end{figure}

Finally we provide the Q-Q plots of a fitted exponential distribution to the excesses.

\begin{figure}[H]
    \centering
    \subfigure[]{\includegraphics[scale=0.14]{ 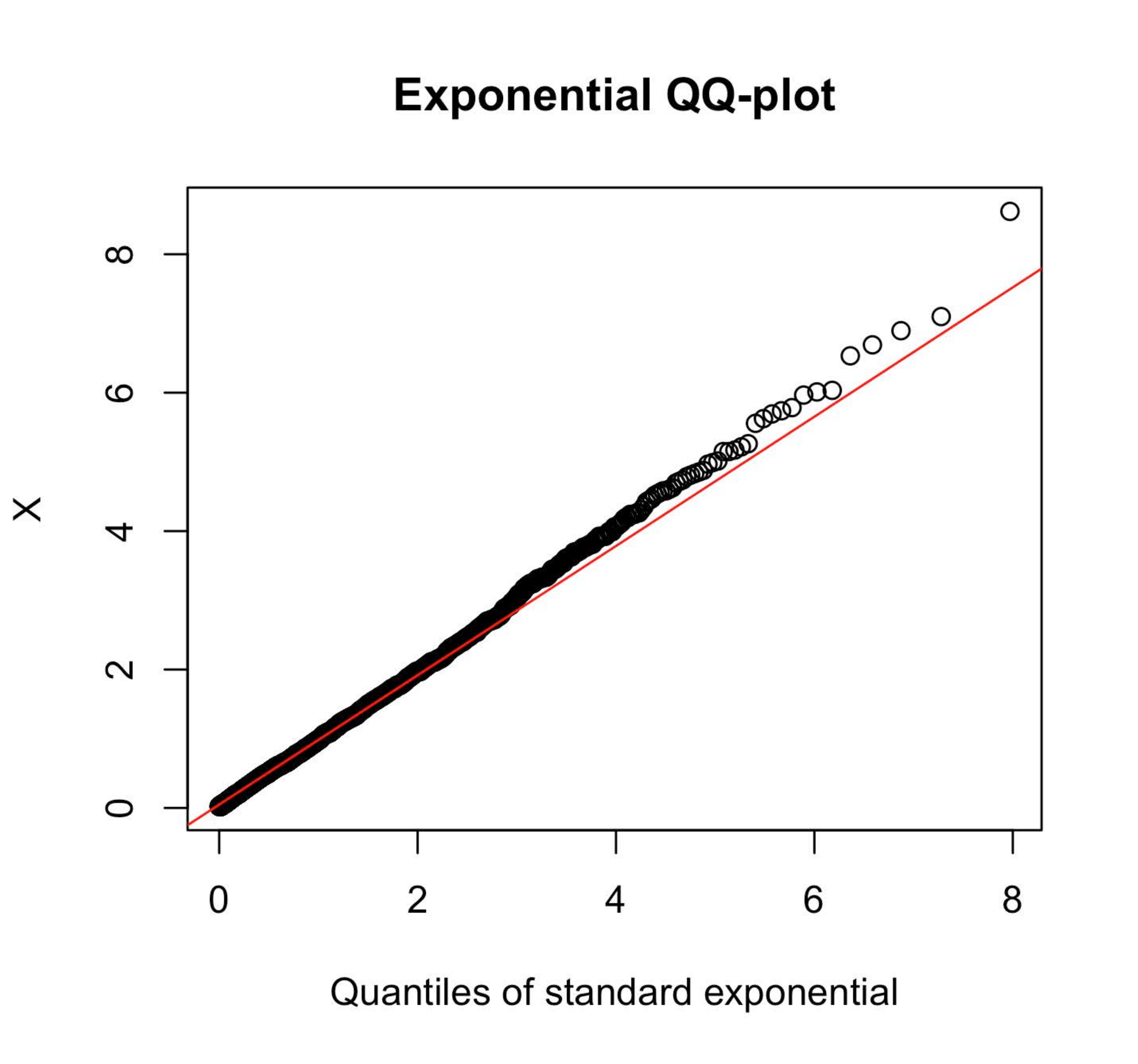}} \subfigure[]{\includegraphics[scale=0.14]{ 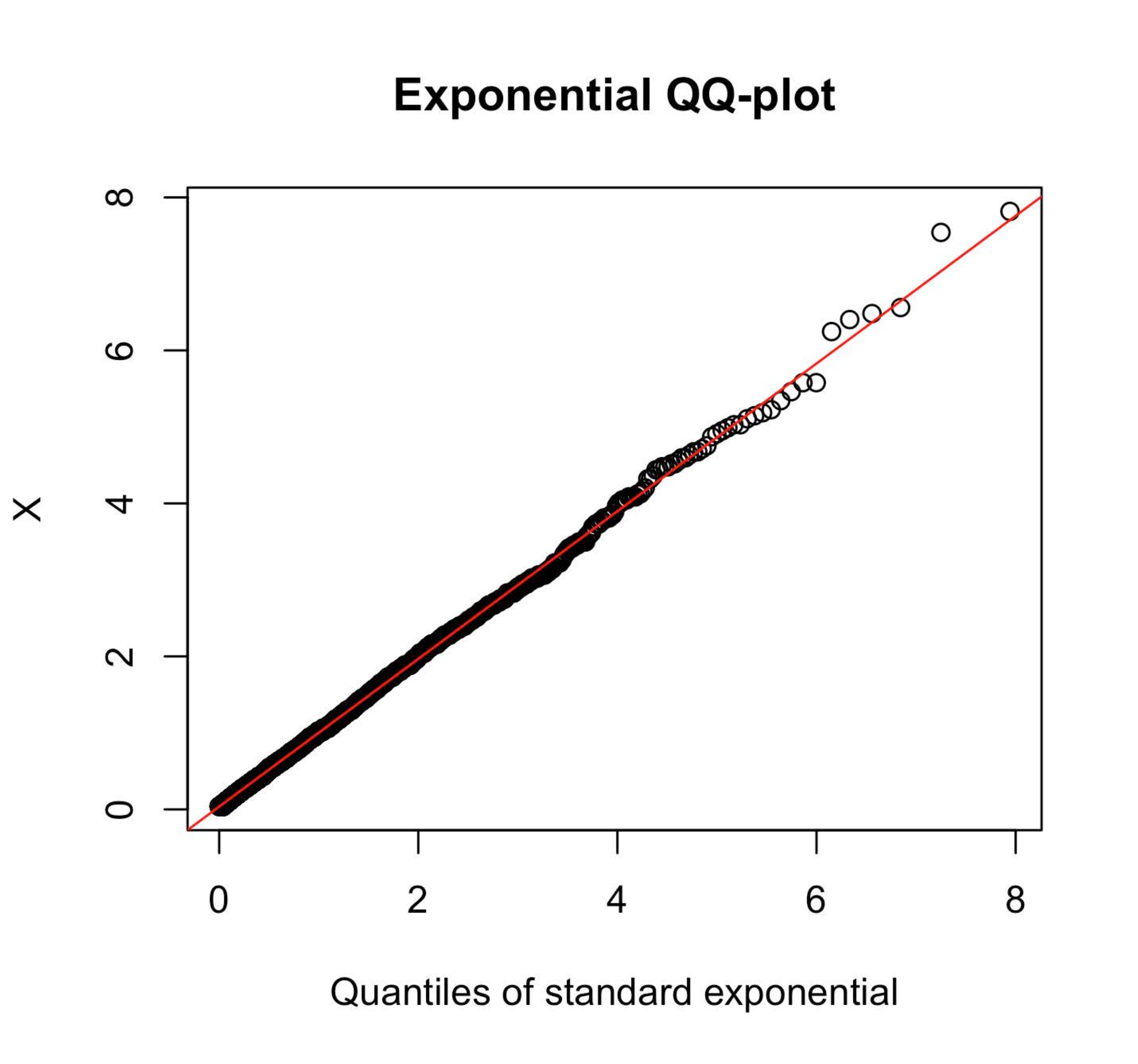}}
    \caption{{\footnotesize Q-Q plots for exponential fitting to the excesses of $\tau$ under case A (a) case B (b). Once more, the data get very close to the identity lines and, by comparing these plots to those from the GPD fit, we note that the approximation based on a GPD and the one using an exponential distribution seem to be both equally as good.}}
    \label{fig:QQExpoEx4and5}
\end{figure}

\subsubsection{Cases when \texorpdfstring{$\theta\leq \ln(2)$}{notheta}}\label{notheta}

We repeat all previous six examples choosing parameters such that $\theta\leq \ln(2)$. As we shall see, this results in a faster extinction, a situation in which the functions from ISMEV used in the fitting of the GPD cannot be applied.

The descriptive statistics along with the distributions used to generated this new set of examples are given in the table below and illustrated using scatter plots. It is noteworthy how the variances are always smaller than 1.

\begin{table}[H]
    \centering
    {\footnotesize
    \begin{tabular}{|c|c|c|c|c|c|c|}
    \hline
       Example  & dist. $f_1$ & dist. $m_1$ & $\theta$ & est. mean & est. variance & range  \\
       \hline
        1A & Bin(3,0.1) & Bin(3,0.1) & 0.6321631 & 1.079 & 0.085 & $\{1,\dots,6\}$\\
        1B & Bin(10,0.032) & Bin(10,0.032) & 0.6504638 & 1.083 & 0.080 & $\{1,\dots,6\}$\\
        2A & Geo(0.9) & Poi(0.55) & 0.6553605 & 1.044 & 0.046 & $\{1,\dots,4\}$\\
        2B & Geom(0.7) & Poi(0.1) & 0.55667 & 1.057 & 0.061 & $\{1,\dots,4\}$\\
        3A & Bin(10,0.039) & Bin(10,0.039) & 0.6504638 & 1.08118 & 0.1129073 & $\{1,\dots,8\}$\\
        3B & Geom(0.6) & Poi(0.1) & 0.61082 & 1.03879 & 0.03990573 & $\{1,\dots,4\}$\\
        \hline
    \end{tabular}
    \caption{Summary of parameters and distributions used in this new set of examples.}}
    \label{tablesummary}
\end{table}

\begin{figure}[H]
\centering
   \begin{tabular}{cc}
       \includegraphics[scale=0.4]{ 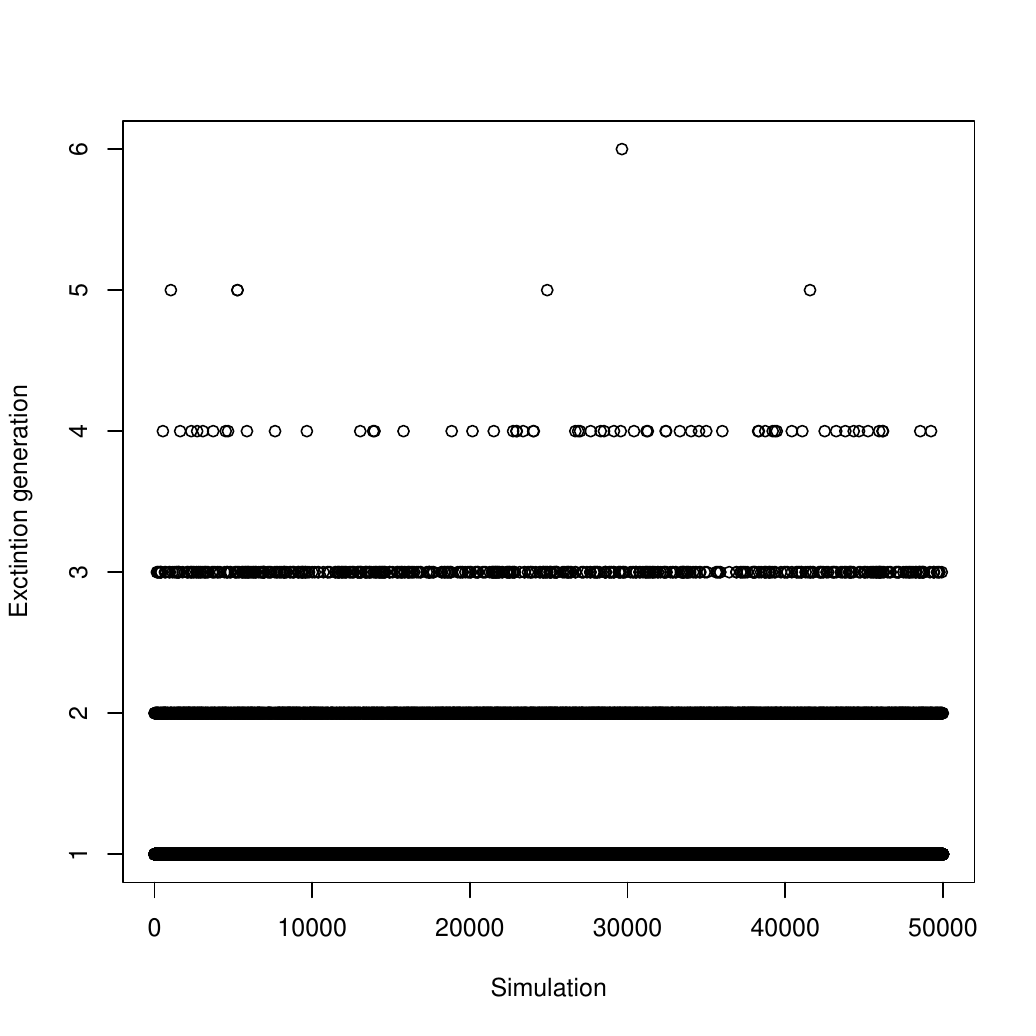} &
    \includegraphics[scale=0.4]{ 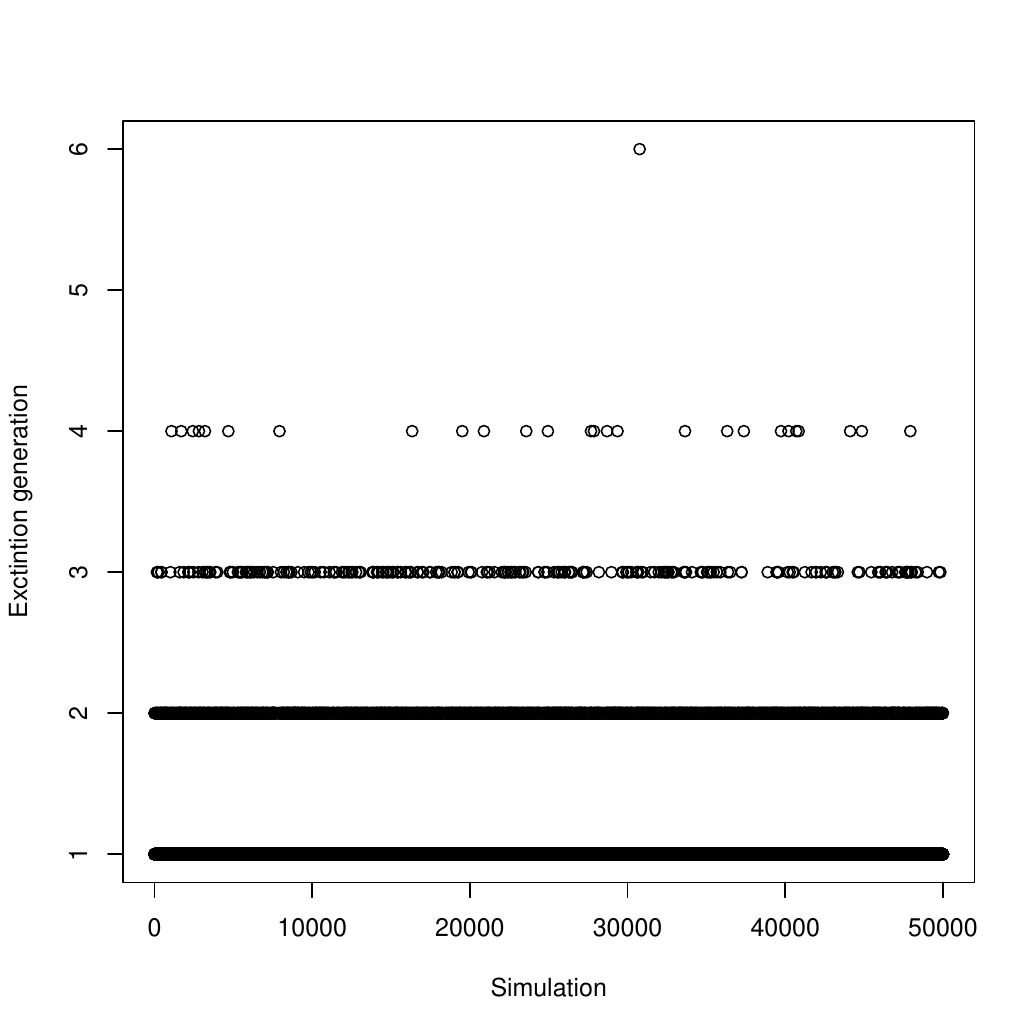}    \\
       (a) & (b) 
       \end{tabular}
       \end{figure}
       \begin{figure}[H]
\centering
   \begin{tabular}{cc}
       \includegraphics[scale=0.4]{ 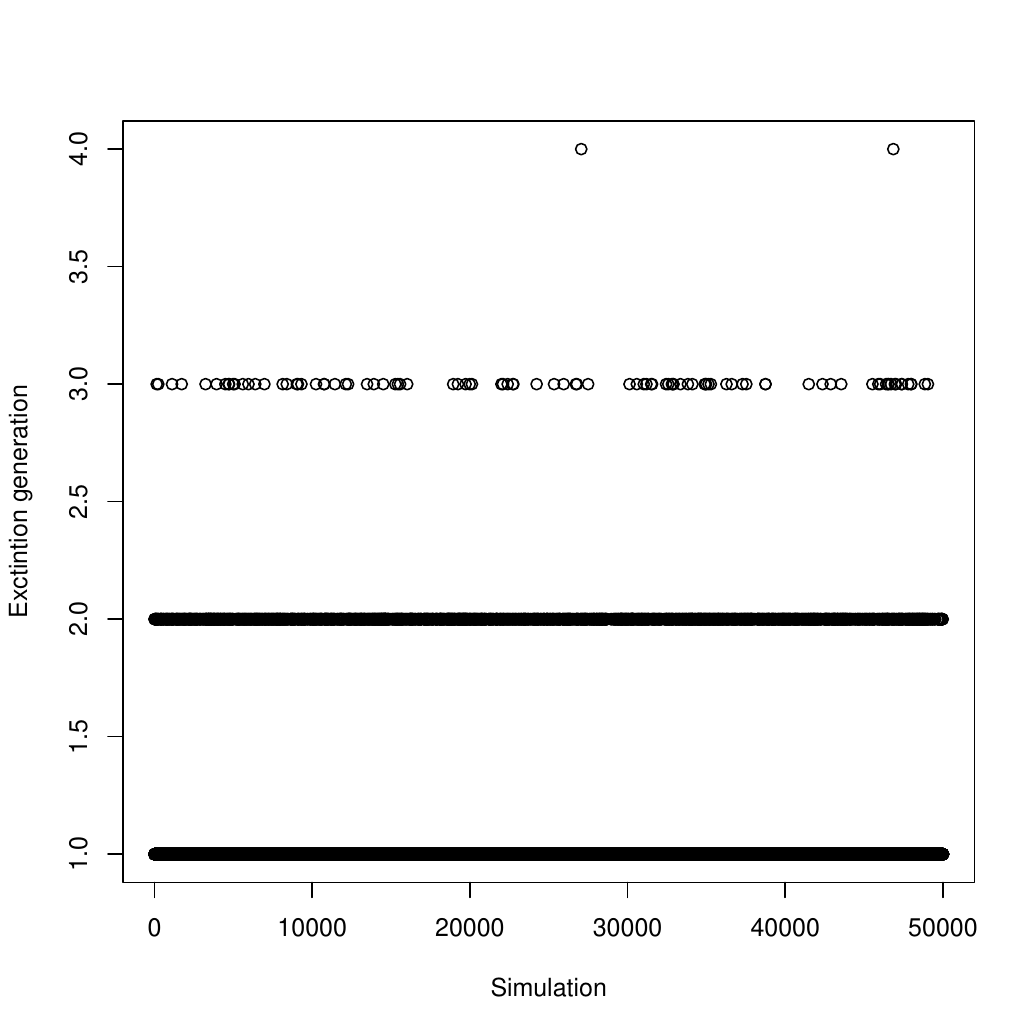} &
    \includegraphics[scale=0.4]{ 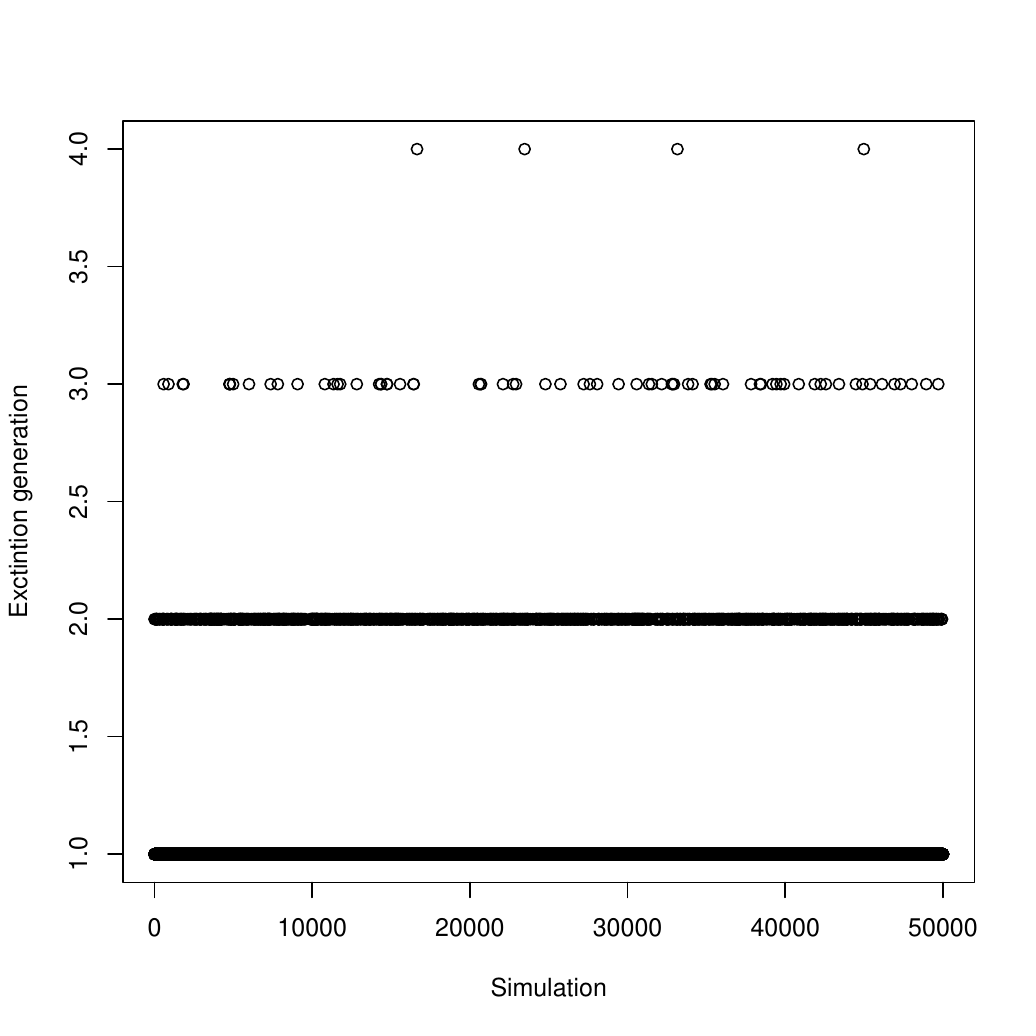}    \\
       (c) & (d) \\
       \includegraphics[scale=0.4]{ 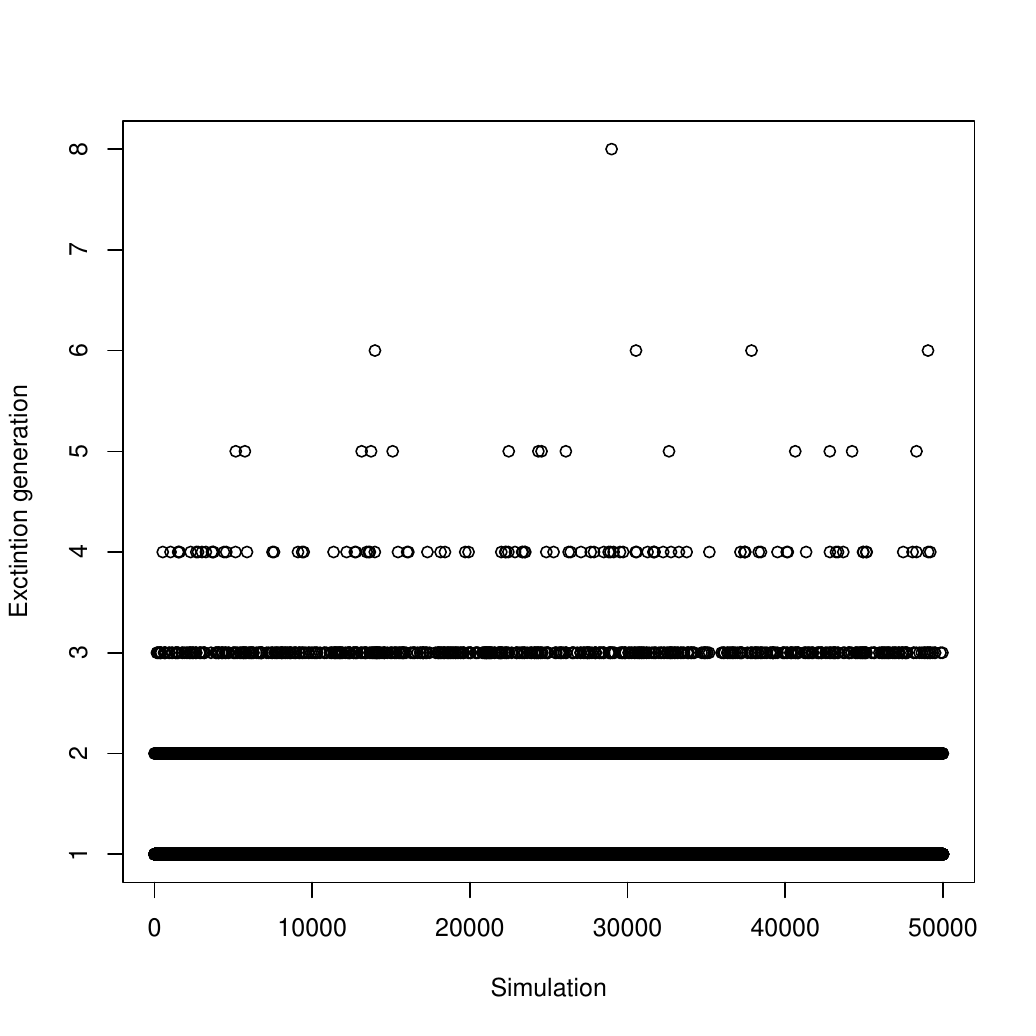} &
    \includegraphics[scale=0.4]{ 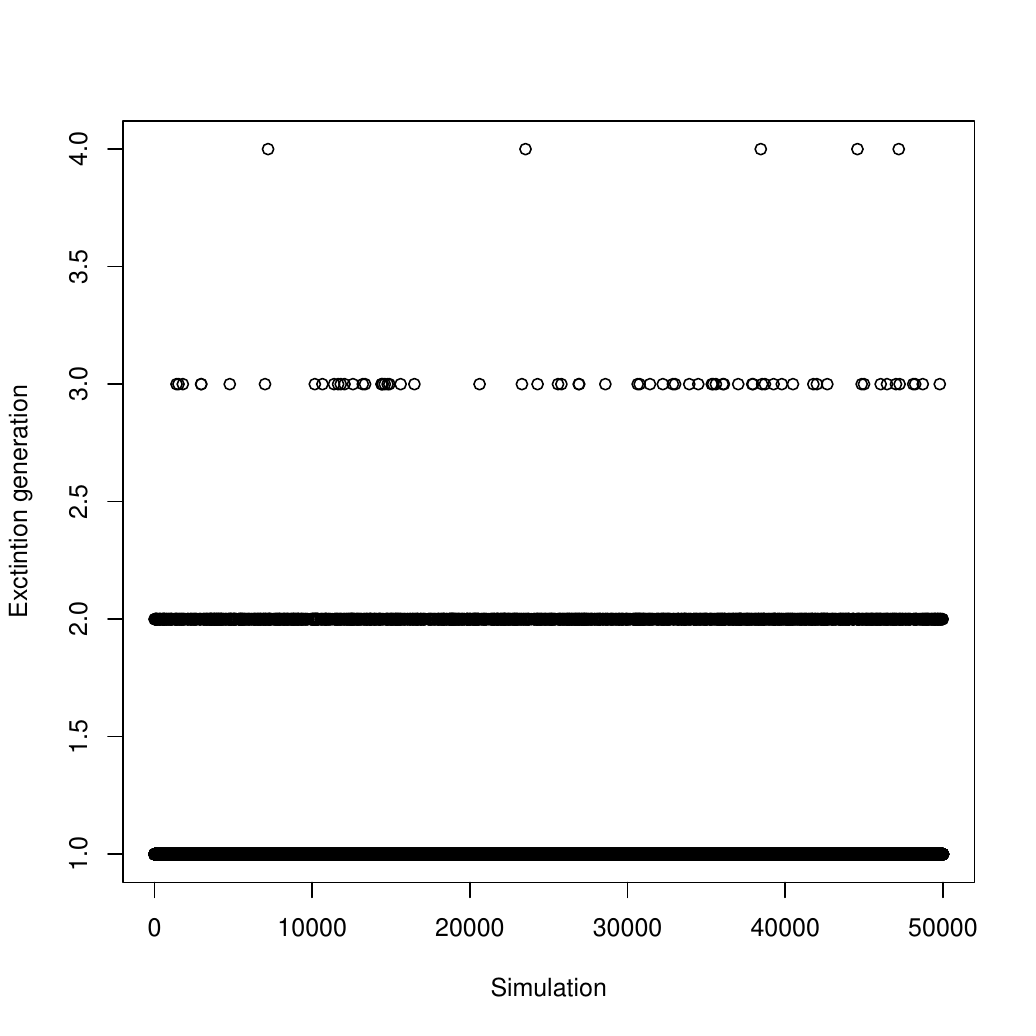}    \\
       (e) & (f) 
   \end{tabular}
    \caption{{\footnotesize Plots of simulated values of $\tau$ when $\theta\leq\ln(2)$, for examples 1A (a), 1B (b), 2A (c), 2B (d), 3A (e) and 3B (f). In all cases we note that the range has decreased significantly and now extinction occurs in a very short time. Hence we cannot apply Theorem 2 in \cite{anderson}, which requires that the discrete distribution has a wide range of values. }}
    \label{fig:enter-label}
\end{figure}
\newpage
Finally, we take a look at the probability $\p\ci \{F_{k+1}=0\}\cup\{M_{k+1}=0\} | Z_k=a\cd$ for $a\geq 1$.

\begin{figure}[H]
    \centering
    \includegraphics[scale=0.5]{ 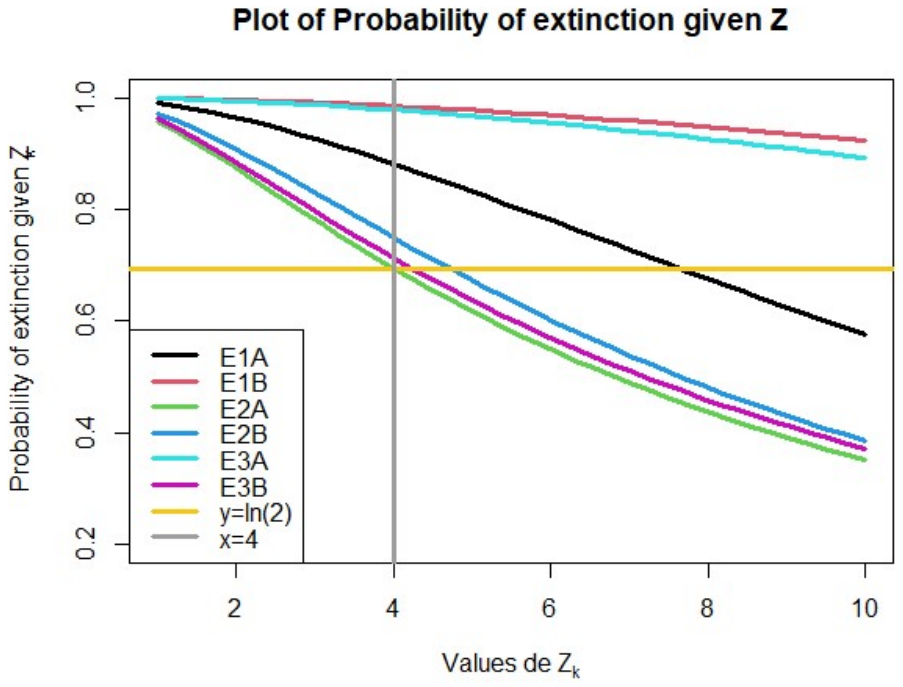}
    \caption{{\footnotesize Plot of $\p\ci \{F_{k+1}=0\}\cup\{M_{k+1}=0\} | Z_k=a\cd$. We note that the value of $\ln(2)$ seems to act as a lower bound for quick extinction. In fact, comparing these to the previously given scatter plots, we see that in all cases the extinction time occurred with the highest frequencies for values between 1 and 4, even in Example 3A when the range is wider than in the other examples.}}
    \label{fig:plotofH}
\end{figure}

\section{Conclusions and further work}

We have presented a formal proof of how, conditioned on almost sure extinction, the distribution of the time to extinction for BGWBP with mating depending on the previous generation, can be approximated using the POT method from EVT. This approximation is very simple and easy to implement, unlike other commonly used approximations in probability theory.

Through our numerical examples we saw the influence of the condition $\theta>\ln(2)$ on the application of the POT method.

A few open problems that result from this work are listed below.

\begin{itemize}
    \item Study of the case when the distributions of $f_1$ and $m_1$ have regularly varying tails and the case of extinction with probability less than one.

    \item A deeper analysis and rigorous proofs of the influence of the condition $\theta>ln(2)$ on the speed of extinction.

    \item Extension of this results to generalizations of this model. For instance the one studied in \cite{bansayeetal}.
\end{itemize}

\textbf{Statements and declarations}

Conflict of interest: The authors declare no conflict of interest.

Funding: The authors declare that this research did not receive any specific grant from funding agencies in the public, commercial, or not-for-profit sectors.

\end{document}